\documentclass[aps,11pt,twoside]{revtex4}
\usepackage{amsmath,latexsym,amssymb,verbatim,enumerate,graphicx}

\usepackage{hyperref}

\usepackage{theorem}
\newtheorem{definition}{Definition}[section]
\newtheorem{proposition}[definition]{Proposition}
\newtheorem{lemma}[definition]{Lemma}

\newtheorem{theorem}{Theorem}
\newtheorem{corollary}[definition]{Corollary}

\def\squareforqed{\hbox{\rlap{$\sqcap$}$\sqcup$}}
\def\qed{\ifmmode\squareforqed\else{\unskip\nobreak\hfil
\penalty50\hskip1em\null\nobreak\hfil\squareforqed
\parfillskip=0pt\finalhyphendemerits=0\endgraf}\fi}
\def\endenv{\ifmmode\;\else{\unskip\nobreak\hfil
\penalty50\hskip1em\null\nobreak\hfil\;
\parfillskip=0pt\finalhyphendemerits=0\endgraf}\fi}
\newenvironment{proof}[1][Proof]{\noindent \textbf{{#1~} }}{\qed}

\newcommand{\bra}[1]{\langle #1|}
\newcommand{\ket}[1]{|#1\rangle}
\newcommand{\braket}[2]{\langle #1|#2\rangle}
\newcommand{\tr}{\text{tr}}

\newcommand{\id}{\mathbb{I}}

\mathchardef\ordinarycolon\mathcode`\:
\mathcode`\:=\string"8000
\def\vcentcolon{\mathrel{\mathop\ordinarycolon}}
\begingroup \catcode`\:=\active
  \lowercase{\endgroup
  \let :\vcentcolon
  }

\newcommand{\nc}{\newcommand}
\nc{\rnc}{\renewcommand} \nc{\beq}{\begin{equation}}
\nc{\eeq}{{\end{equation}}} \nc{\bea}{\begin{eqnarray}}
\nc{\eea}{\end{eqnarray}} \nc{\beqa}{\begin{eqnarray}}
\nc{\eeqa}{\end{eqnarray}} \nc{\lbar}[1]{\overline{#1}}

\nc{\conv}{\operatorname{conv}}
\nc{\smfrac}[2]{\mbox{$\frac{#1}{#2}$}} \nc{\Tr}{\operatorname{Tr}}
\nc{\ox}{\otimes} \nc{\dg}{\dagger} \nc{\dn}{\downarrow}
\nc{\lmax}{\lambda_{\text{max}}}
\nc{\lmin}{\lambda_{\text{min}}}

\nc{\cA}{{\cal A}} \nc{\cB}{{\cal B}} \nc{\cC}{{\cal C}}
\nc{\cD}{{\cal D}} \nc{\cE}{{\cal E}} \nc{\cF}{{\cal F}}
\nc{\cG}{{\cal G}} \nc{\cH}{{\cal H}} \nc{\cI}{{\cal I}}
\nc{\cJ}{{\cal J}} \nc{\cK}{{\cal K}} \nc{\cL}{{\cal L}}
\nc{\cM}{{\cal M}} \nc{\cN}{{\cal N}} \nc{\cO}{{\cal O}}
\nc{\cP}{{\cal P}} \nc{\cQ}{{\cal Q}} \nc{\cR}{{\cal R}} \nc{\cS}{{\cal S}}
\nc{\cT}{{\cal T}} \nc{\cU}{{\cal U}} \nc{\cV}{{\cal V}}
\nc{\cX}{{\cal X}} \nc{\cW}{{\cal W}} \nc{\cZ}{{\cal Z}}

\nc{\CA}{{\cal A}} \nc{\CB}{{\cal B}} \nc{\CC}{{\cal C}}
\nc{\CD}{{\cal D}} \nc{\CE}{{\cal E}} \nc{\CF}{{\cal F}}
\nc{\CG}{{\cal G}} \nc{\CH}{{\cal H}} \nc{\CI}{{\cal I}}
\nc{\CJ}{{\cal J}} \nc{\CK}{{\cal K}} \nc{\CL}{{\cal L}}
\nc{\CM}{{\cal M}} \nc{\CN}{{\cal N}} \nc{\CO}{{\cal O}}
\nc{\CP}{{\cal P}} \nc{\CQ}{{\cal Q}} \nc{\CR}{{\cal R}} \nc{\CS}{{\cal S}}
\nc{\CT}{{\cal T}} \nc{\CU}{{\cal U}} \nc{\CV}{{\cal V}}
\nc{\CX}{{\cal X}} \nc{\CW}{{\cal W}} \nc{\CZ}{{\cal Z}}

\nc{\csupp}{{\operatorname{csupp}}}
\nc{\qsupp}{{\operatorname{qsupp}}} \nc{\var}{\operatorname{var}}
\nc{\rar}{\rightarrow} \nc{\lrar}{\longrightarrow}
\nc{\poly}{\operatorname{poly}}
\nc{\polylog}{\operatorname{polylog}} \nc{\Lip}{\operatorname{Lip}}
\nc{\mb}[1]{\mathbf{#1}}
\nc{\ep}{\epsilon}
\nc{\Om}{\Omega}
\nc{\wt}[1]{\widetilde{#1}}

\def\>{\rangle}
\def\<{\langle}

\nc{\glneq}{{\raisebox{0.6ex}{$>$}  \hspace*{-1.8ex} \raisebox{-0.6ex}{$<$}}}
\nc{\gleq}{{\raisebox{0.6ex}{$\geq$}\hspace*{-1.8ex} \raisebox{-0.6ex}{$\leq$}}}

\nc{\RR}{{{\mathbb R}}}
\nc{\FF}{{{\mathbb F}}}
\nc{\HH}{{{\mathbb H}}}
\nc{\NN}{{{\mathbb N}}}
\nc{\ZZ}{{{\mathbb Z}}}
\nc{\PP}{{{\mathbb P}}}
\nc{\QQ}{{{\mathbb Q}}}
\nc{\UU}{{{\mathbb U}}}
\nc{\WW}{{{\mathbb W}}}
\nc{\EE}{{{\mathbb E}}}
\rnc{\SS}{{{\mathbb S}}}

\nc{\vholder}[1]{\rule{0pt}{#1}}
\nc{\wh}[1]{\widehat{#1}}
\nc{\h}[1]{\widehat{#1}}

\nc{\ob}[1]{#1}

\def\beq{\begin {equation}}
\def\eeq{\end {equation}}

\nc{\eq}[1]{Eq.~(\ref{eq:#1})} \nc{\eqs}[2]{Eqs.~(\ref{eq:#1}) and
(\ref{eq:#2})}

\nc{\eqn}[1]{Eq.~(\ref{eqn:#1})}
\nc{\eqns}[2]{Eqs.~(\ref{eqn:#1}) and (\ref{eqn:#2})}

\nc{\region}{\cS\cW}

\begin{document}

\title{{\Large A Generalization of Quantum Stein's Lemma}}

\author{Fernando G.S.L. Brand\~ao}
\email{fernando.brandao@imperial.ac.uk}
\affiliation{Institute for Mathematical Sciences, Imperial
College London, London SW7 2BW, UK}
\affiliation{QOLS, Blackett Laboratory, Imperial College
London, London SW7 2BW, UK}

\author{Martin B. Plenio}
\email{m.plenio@imperial.ac.uk}
\affiliation{Institute for Mathematical Sciences, Imperial
College London, London SW7 2BW, UK}
\affiliation{QOLS, Blackett Laboratory, Imperial College
London, London SW7 2BW, UK}


\begin{abstract}

Given many independent and identically-distributed (i.i.d.) copies of 
a quantum system described either by the state $\rho$ or $\sigma$ (called null 
and alternative hypotheses, respectively), what is the optimal measurement 
to learn the identity of the true state? In asymmetric 
hypothesis testing one is interested in minimizing the probability of mistakenly 
identifying $\rho$ instead of $\sigma$, while requiring that the 
probability that $\sigma$ is identified in the place of $\rho$ is bounded 
by a small fixed number. Quantum Stein's Lemma identifies the asymptotic 
exponential rate at which the specified error probability tends to zero 
as the quantum relative entropy of $\rho$ and $\sigma$.

We present a generalization of quantum Stein's Lemma to the situation in 
which the alternative hypothesis is formed by a family of states, which 
can moreover be non-i.i.d.. We consider sets of states which satisfy a 
few natural properties, the most important being the closedness under 
permutations of the copies. We then determine the error rate function in a very similar fashion to quantum Stein's 
Lemma, in terms of the quantum relative entropy. 

Our result has two applications to entanglement theory. First it gives an operational 
meaning to an entanglement measure known as regularized relative entropy 
of entanglement. Second, it shows that this measure is faithful, being 
strictly positive on every entangled state. This implies, in particular, 
that whenever a multipartite state can be asymptotically converted into 
another entangled state by local operations and classical communication, 
the rate of conversion must be non-zero. Therefore, the operational 
definition of multipartite entanglement is equivalent to its mathematical 
definition.

\end{abstract}

\maketitle

\parskip .75ex


\section{Introduction}

Hypothesis testing refers to a general set of tools in statistics and 
probability theory for making decisions based on experimental data from 
random variables. In a typical scenario, an experimentalist is faced 
with two possible hypotheses and must decide based on experimental 
observation which one was actually realized. There are two types 
of errors in this process, corresponding to mistakenly identifying 
one of the two options when the other should have been detected. A 
central task in hypothesis testing is the development of optimal 
strategies for minimizing such errors and the determination of compact 
formulae for the minimum error probabilities. 

Substantial progress has been achieved both in the classical and quantum 
settings for i.i.d processes \cite{CT91, Che52, CL71, Bla74, HP91, ON00, Hay02, OH04, NS06, ACM+07, Nag06, NH07, ANSV07, Hay07}. The non-i.i.d. case, however, has proven 
harder and much less is known. The main result of this paper is a 
particular instance of quantum hypothesis testing of non-i.i.d. sources 
for which the optimal separation rate can be fully determined. To the 
best of the authors knowledge, the complete solution of such a problem was not 
known even in the classical case. 

Suppose we have access to a source that generates independent and 
identically-distributed random variables according to one of two 
possible probability distributions. Our aim is to decide which 
probability distribution is the true one. In the quantum generalization 
of the problem, we are faced with a source that emits several i.i.d. 
copies of one of two quantum states $\rho$ and $\sigma$, and we should 
decide which of them is being produced. Since the quantum setting also 
encompasses the classical, we will focus on the former.

In order to learn the identity of the state the observer measures a two 
outcome POVM $\{ A_n, \id - A_n \}$ given $n$ realizations of the unknown 
state. If he obtains the outcome associated to $A_n$ ($\id - A_n$) then 
he concludes that the state was $\rho$ ($\sigma$). The state $\rho$ is 
seen as the null hypothesis, while $\sigma$ is the alternative hypothesis. 
There are two types of errors: 
\begin{itemize}
        \item Type I: The observer finds that the state was $\sigma$, when 
        in reality it was $\rho$. This happens with probability 
        $\alpha_n(A_n) := \tr(\rho^{\otimes n}(\id - A_n))$.
        \item Type II: The observer finds that the state was $\rho$, when 
        it actually was $\sigma$. This happens with probability 
        $\beta_n(A_n) := \tr(\sigma^{\otimes n}A_n)$.
\end{itemize}
There are several distinct settings that might be considered, depending 
on the importance we attribute to the two types of errors 
\cite{CT91, Che52, CL71, Bla74, HP91, ON00, Hay02, OH04, NS06, ACM+07, Nag06, NH07, ANSV07, Hay07}. 

In \textit{asymmetric hypothesis testing}, the probability 
of type II error should be minimized to the extreme, while only requiring 
that the probability of type I error is bounded by a small parameter 
$\epsilon$. The relevant error quantity in this case can be written as
\begin{equation}  
        \beta_n(\epsilon) := \min_{0 \leq A_n \leq \id} \{ \beta_n(A_n) : 
        \alpha_n(A_n) \leq \epsilon \}.
\end{equation}
Quantum Stein's Lemma \cite{HP91, ON00} states that for every $0 < \epsilon < 1$,
\begin{equation} \label{Stein}
\lim_{n \rightarrow \infty} - \frac{\log(\beta_n(\epsilon))}{n} = S(\rho || \sigma). 
\end{equation}
where $S(\rho || \sigma) = \tr(\rho(\log(\rho) - \log(\sigma)))$ is the quantum relative entropy (or quantum Kullback-Leibler divergence) 
of $\rho$ and $\sigma$. This fundamental result gives a rigorous operational interpretation for 
the quantum relative entropy and was proven by Hiai and Petz \cite{HP91} 
and Ogawa and Nagaoka \cite{ON00}. Different proofs have since be given 
in Refs. \cite{Hay02, OH04, ANSV07}. The relative entropy is also the 
asymptotic optimal exponent for the decay of $\beta_n$ when we require 
that $\alpha_n \stackrel{n \rightarrow \infty}{\longrightarrow} 0$ 
\cite{OH04}. 

Quantum Stein's Lemma can be generalized in two natural directions. We 
can consider asymmetric hypothesis testing of \textit{non-i.i.d.} states 
and, moreover, we can allow the two hypotheses to be composed of sets 
of states, instead of a single one. In this more general formulation, 
the problem cannot be solved in simple terms as in quantum Stein's Lemma. 
It is an interesting line of investigation, therefore, to study under 
what further assumptions the optimal error exponent can be determined 
in an illustrative manner.     

There are several works that present extensions of quantum Stein's Lemma. 
Concerning non-i.i.d. sequences, in \cite{BS04} Bjelakovi\'c and Siegmund-Schultze proved that quantum Stein's Lemma is also true if 
the null hypothesis is an ergodic state, instead of i.i.d.. Further generalizations to particular cases where the null and alternative hypotheses are correlated states were obtained in Refs. 
\cite{BDK+08, HMO07, MHOF08}. Finally, the \textit{information spectrum} approach 
\cite{NH07} delivers the achievability and strong converse optimal rate 
limits in terms of divergence spectrum rates for arbitrary sequence of 
states. Despite its generality, this method has the drawback that in general 
no direct connection to the quantum relative entropy is established.  

Concerning extensions to sets of states as hypotheses, a generalization 
of quantum Stein's Lemma, sometimes referred to as quantum Sanov's Theorem, 
considers the situation in which the null hypotheses are i.i.d extensions 
of the elements of a family of states ${\cal K}$ \cite{Hay02, BDK+05}. It 
was found that the rate limit of type II error is given by $\inf_{\rho 
\in {\cal K}}S(\rho || \sigma)$, which is a pleasingly direct extension 
of the original result. In Ref. \cite{BDK+08} generalizations to the case of 
correlated families of states as the null hypothesis were presented. 

The main result of this paper has a similar flavor to the above-mentioned 
generalizations. We will however be interested in the case where the 
\textit{alternative hypothesis} is not only composed of a single i.i.d. 
state, but is actually formed by a family of non-i.i.d. states satisfying 
certain conditions to be specified in the next section. We will then show 
that the regularization of the minimum quantum relative entropy over the 
set of states considered is the optimal rate limit for type II error. 

Apart from extending the range of possibilities of the alternative 
hypothesis, instead of the null hypothesis, the present work differs 
from previous ones in the assumptions which are imposed on the set of 
states. Instead of ergodicity and related ideas, we consider as the 
alternative hypothesis sets of states satisfying five properties outlined 
in section \ref{mainr}, the most important being the closedness under 
the permutations of the copies of the state. In this way, we will be 
able to employ recent advances in the characterization of quantum 
permutation-invariant states, more specifically the exponential de 
Finetti Theorem due to Renner \cite{Ren05, Ren07}, to reduce the problem 
from the most general form to particular one closely related to the i.i.d., in which it can be tackled
more easily. 

The main motivation for considering these particular sets of states comes 
from entanglement theory \cite{PV07, HHHH07}. Given a $k$-partite finite 
dimensional Hilbert space ${\cal H} := {\cal H}_{1} \otimes ... \otimes 
{\cal H}_{k}$, we say that a state $\sigma$ acting on ${\cal H}$ is 
separable if it can be written as 
\begin{equation} \label{separable}
        \sigma = \sum_j p_j \sigma_{1, j} \otimes ... \otimes \sigma_{k, j},
\end{equation}
for local states $\sigma_{i, j} \in {\cal D}({\cal H}_i)$ and a probability
distribution $\{ p_j \}$ \cite{Wer89}. Assuming that the state $\sigma$ is 
shared by $k$ parties, each holding a quantum system described by the 
Hilbert space ${\cal H}_j$, it is clear that they can generate it from 
a completely uncorrelated state by \textit{local quantum operations} 
on their respective particles and \textit{classical communication} among 
them (LOCC). If a state cannot be created by LOCC, we say it is 
\textit{entangled}. To create an entangled state from an uncorrelated 
state the parties must, in addition to LOCC, exchange quantum particles. 
As we show, the set of separable states satisfy the conditions we impose 
on the alternative hypothesis. Therefore, a particular instance of the 
problem we analyse is the discrimination of tensor powers of an entangled state 
from an arbitrary sequence of separable states. 


\textbf{Notation:} We let ${\cal H}$ be a finite dimensional Hilbert 
space and ${\cal D}({\cal H})$ the set of density operators acting on 
${\cal H}$. Given a pure state $\ket{\theta} \in {\cal H}$, 
${\cal H} \bot \ket{\theta}$ denotes the subspace of ${\cal H}$ orthogonal 
to $\ket{\theta}$. Let $\text{supp}(X)$ be the support of the operator $X$. For two states $\rho, \sigma \in {\cal D}({\cal H})$ with $\text{supp}(\rho) \subseteq \text{supp}(\sigma)$, we define the quantum relative entropy of $\rho$ and $\sigma$ as 
\begin{displaymath}
        S(\rho || \sigma) := \tr(\rho(\log(\rho) - \log(\sigma))).
\end{displaymath}

Given a Hermitian operator $A$, $||A||_1 = \tr(\sqrt{A^{\cal y}A})$ 
stands for the trace norm of $A$, $\tr(A)_+$ for the trace of the 
positive part of $A$, i.e. the sum of the positive eigenvalues of $A$, and $\lambda_{\max}(A)$ and $\lambda_{\min}(A)$ for the maximum and the minimum eigenvalue of $A$, respectively. For two positive semidefinite operators $A, B$, 
$F(A, B) := \tr(\sqrt{A^{1/2}BA^{1/2}})$ is their fidelity. The partial trace of $\rho \in {\cal D}({\cal H}^{\otimes n})$ with respect 
to the $j$-th Hilbert space is denoted by $\tr_{j}(\rho)$, while 
$\tr_{\backslash j}(\rho)$ stands for the partial trace of all Hilbert 
spaces, except the $j$-th. We denote the binary Shannon entropy by $h$: $h(x) = -x\log(x) - (1 - x)\log(1 - x)$.

Given a subset ${\cal M} \subseteq \mathbb{R}^{n}$ 
we define its associate cone by $\text{cone}({\cal M}) := \{ x : 
x = \lambda y, y \in {\cal M}, \lambda \in \mathbb{R}_+  \}$ and its 
dual cone by ${\cal M}^* := \{ x : y^{T}x \geq 0 \hspace{0.1 cm} 
\forall \hspace{0.05 cm} y \in {\cal M} \}$. We denote the $\epsilon$-ball 
in trace norm around $\rho$ by $B_{\epsilon}(\rho) := \{ \pi \in {\cal D}({\cal H}) : ||\rho - \pi ||_1 \leq \epsilon \}$. 
The Bachmann-Landau notation $g(n) = O(f(n))$ stands for $\exists k > 0, n_0 : \forall n 
> n_0, \hspace{0.1 cm} g(n) \leq k f(n)$, while $g(n) = o(f(n))$ for 
$\forall k > 0, \exists n_0 : \forall n > n_0, \hspace{0.1 cm} g(n) 
\leq k f(n)$.  

A function $E$ is called asymptotically continuous 
if there is a monotonic increasing function $f : \mathbb{R} \rightarrow \mathbb{R}$ satisfying $\lim_{x \rightarrow 0^+}f(x)=0$ such that $\forall \rho,\sigma\in 
{\cal D}({\cal H})$,  $|E(\rho)-E(\sigma)|\leq \log(\dim({\cal H})) f(||\rho - \sigma||_1)$.

Let $\text{Sym}({\cal H}^{\otimes n})$ denote the symmetric subspace of ${\cal H}^{\otimes n}$. For any $\ket{\psi} \in {\cal H}^{\otimes n}$ not orthogonal to $\text{Sym}({\cal H}^{\otimes n})$, we define 
\begin{equation}
        \text{Sym}(\ket{\psi}) := \frac{\sum_{\pi \in S_n} P_{\pi}\ket{\psi}}{\left \Vert \sum_{\pi \in S_n} P_{\pi}\ket{\psi} \right \Vert} 
\end{equation}
where $S_n$ is the symmetric group of order $n$ and $P_{\pi}$ is the representation in ${\cal H}^{\otimes n}$ of a permutation $\pi \in S_n$ given by $P_{\pi} \left(\psi_1 \otimes \psi_2 \otimes ... \otimes \psi_n \right) = \psi_{\pi^{-1}(1)} \otimes \psi_{\pi^{-1}(2)} \otimes ... \otimes \psi_{\pi^{-1}(n)}$. Finally, the symmetrization superoperator $\hat{S}_n : {\cal B}({\cal H}^{\otimes n}) \rightarrow {\cal B}({\cal H}^{\otimes n})$ is defined as
\begin{equation}
        \hat{S}_n(X) := \frac{1}{n!} \sum_{\pi \in S_n} P_\pi X  P_\pi^{*} .
\end{equation}

\section{Definitions and Main Results} \label{mainr}

Given a set of states ${\cal M} \subseteq {\cal D}({\cal H})$ we define 
\begin{equation} \label{relent1}
        E_{{\cal M}}(\rho) := \inf_{\sigma \in {\cal M}} 
        S(\rho || \sigma),
\end{equation}
and
\begin{equation}
        LR_{\cal M}(\rho) := \inf_{\sigma \in {\cal M}} 
        S_{\max}(\rho || \sigma),
\end{equation}
where
\begin{equation}
        S_{\max}(\rho || \sigma) := \inf \{ s : \rho \leq 2^s \sigma   \} 
\end{equation}
is the maximum relative entropy \cite{Dat08a}. Note 
that if we take ${\cal M}$ to be the set of separable states, $E_{\cal M}$ 
and $LR_{\cal M}$ reduce to two entanglement measures known as the 
relative entropy of entanglement \cite{VPRK97, VP98} and the logarithm 
global robustness of entanglement \cite{VT99, HN03, Bra05, Dat08b}. This 
connection is the reason for the nomenclature used here. 

We will also need the smooth version of $LR_{\cal M}$, defined as
\begin{equation} \label{smooth}
        LR_{{\cal M}}^{\epsilon}(\rho) := \min_{\tilde{\rho} \in 
        B_{\epsilon}(\rho)} LR_{\cal M}(\tilde{\rho}).
\end{equation}
We note that smooth 
versions of other non-asymptotic-continuous measures, such as the 
min- and max-entropies \cite{RW04, Ren05, MPB08}, have been 
proposed and shown to be useful in non-asymptotic and non-i.i.d. 
information theory.

Let us specify the sets of states over which the alternative hypothesis 
can vary. We will consider any family of sets 
$\{ {\cal M}_n \}_{n \in \mathbb{N}}$, with 
${\cal M}_n \subseteq {\cal D}({\cal H}^{\otimes n})$, satisfying the following 
properties
\begin{enumerate}
        \item \label{cond1} Each ${\cal M}_n$ is convex and closed. 
        \item \label{cond2} Each ${\cal M}_n$ contains $\sigma^{\otimes n}$, 
        for a full rank state $\sigma \in {\cal D}({\cal H})$.
        \item \label{cond3} If $\rho \in {\cal M}_{n + 1}$, then 
        $\tr_{k}(\rho) \in {\cal M}_{n}$, for every $k \in \{1, ..., n + 1 \}$.
        \item \label{cond4} If $\rho \in {\cal M}_{n}$ and 
        $\nu \in {\cal M}_m$, then $\rho \otimes \nu \in 
        {\cal M}_{n + m}$.
        \item \label{cond5} If $\rho \in {\cal M}_n$, then 
        $P_{\pi}\rho P_{\pi} \in {\cal M}_n$ for every $\pi \in S_n$.
\end{enumerate}

We define the regularized version of the quantity given by Eq. 
(\ref{relent1}) as 
\begin{equation}  \label{regu1}
        E_{\cal M}^{\infty}(\rho) := \lim_{n \rightarrow \infty} 
        \frac{1}{n} E_{{\cal M}_n}(\rho^{\otimes n}).
\end{equation}
To see that the limit exists in Eq. 
(\ref{regu1}) we use the fact that if a sequence $(a_n)$ satisfies $a_{n + m} \leq a_n + a_m$, 
then $a_n/n$ is convergent (see e.g. Lemma 4.1.2 in \cite{Dav07}). Using property \ref{cond4} 
it is easy to see that our sequence satisfies this condition.

We now turn to the main result of the paper. Suppose we have one of 
the following two hypothesis: 
\begin{enumerate}
        \item \textit{Null hypothesis}: For every $n \in \mathbb{N}$ we 
        have $\rho^{\otimes n}$ with $\rho\in{\cal D}({\cal H})$.
        \item \textit{Alternative hypothesis}: For every $n \in \mathbb{N}$
        we have an unknown state $\omega_n \in {\cal M}_n$, where 
        $\{ {\cal M}_n \}_{n \in \mathbb{N}}$ is a family of sets satisfying 
        properties \ref{cond1}-\ref{cond5}.
\end{enumerate}
The next theorem gives the optimal rate limit for the type II error when 
one requires that type I error vanishes asymptotically.

\begin{theorem} \label{maintheorem}
Let $\{ {\cal M}_n \}_{n \in \mathbb{N}}$ be a family of sets satisfying 
properties \ref{cond1}-\ref{cond5} and $\rho \in {\cal D}({\cal H})$. Then 

(\textit{Direct part}): For every  $\epsilon > 0$ there exists a sequence of POVMs $\{ A_n, \id - 
A_n \}_{n \in \mathbb{N}}$ such that
\begin{equation}
        \lim_{n \rightarrow \infty} \tr((\id - A_n) \rho^{\otimes n}) = 0 
\end{equation}
and for every $n \in \mathbb{N}$ and $\omega_n \in {\cal M}_n$,
\begin{equation}
        - \frac{\log \tr(A_n \omega_n)}{n} + {\epsilon} \geq 
        E_{\cal M}^{\infty}(\rho) .
\end{equation}

(\textit{Strong Converse}): If a real number $\epsilon > 0$ and a sequence of POVMs 
$\{ A_n, \id - A_n \}_{n \in \mathbb{N}}$ are such that for every $n \in \mathbb{N}$ and $\omega_n \in {\cal M}_n$,
\begin{equation}
         - \frac{\log( \tr(A_n \omega_n))}{n} - \epsilon \geq 
         E_{\cal M}^{\infty}(\rho),
\end{equation}
then
\begin{equation}
        \lim_{n \rightarrow \infty} \tr((\id - A_n) \rho^{\otimes n}) = 1. 
\end{equation}
\end{theorem}

We note that the converse part of the theorem is a so called \textit{strong converse}, which shows that not only the probability of type I error does not go to zero when we require that type II error rate is larger than $E_{\cal M}^{\infty}$, but it actually goes to one. 

Also note we can recover the original quantum Stein's Lemma by choosing ${\cal M}_n := \{ \sigma^{\otimes n} \}$, where $\sigma$ is the alternative hypothesis and $\rho$ is the null hypothesis (Theorem \ref{maintheorem} can only be applied here if $\text{supp}(\rho) \subseteq \text{supp}(\sigma)$, but this is exactly the non-trivial case of quantum Stein's Lemma).    

Theorem \ref{maintheorem} gives an operational interpretation to the 
\textit{regularized} relative entropy of entanglement 
\cite{VPRK97, VP98, VW01}, defined by
\begin{equation}
        E_R^{\infty}(\rho) := \lim_{n \rightarrow \infty} \frac{1}{n} 
        \min_{\sigma \in {\cal S}({\cal H}^{\otimes n})} S(\rho^{\otimes n} || \sigma),
\end{equation}
with ${\cal S}({\cal H}^{\otimes n})$ as the set of $k$-partite separable states 
over ${\cal H}^{\otimes n} := {\cal H}_1^{\otimes n} \otimes ... \otimes 
{\cal H}_k^{\otimes n}$, where the $j$-th local party Hilbert space is given by ${\cal H}_{j}^{\otimes n}$. Taking ${\cal M}_n = {\cal S}({\cal H}^{\otimes n})$, 
it is a simple exercise to check that they satisfy conditions 
\ref{cond1}-\ref{cond5}. Therefore, we conclude that $E_R^{\infty}(\rho)$ 
gives the asymptotic rate of type II error when we try to decide if 
we have several realizations of $\rho$ or a sequence of \textit{arbitrary} 
separable states. This rigorously justifies the use of the regularized 
relative entropy of entanglement as a measure of distinguishability 
of quantum correlations from classical correlations, as was originally
suggested on heuristic grounds in \cite{VPJK97,VP98}.

On the way to prove Theorem \ref{maintheorem} we establish the following alternative 
expression for $E_{\cal M}^{\infty}$.  

\begin{proposition} \label{relenteqrob}
For every family of sets  $\{ {\cal M}_n \}_{n \in \mathbb{N}}$ satisfying 
properties \ref{cond1}-\ref{cond5} and every state $\rho \in {\cal D}({\cal H})$,
\begin{equation} \label{relenteqrobeq}
E_{\cal M}^{\infty}(\rho) = \lim_{\epsilon \rightarrow 0} \liminf_{n \rightarrow \infty} \frac{1}{n} LR_{{\cal M}_n}^{\epsilon}(\rho^{\otimes n}) = \lim_{\epsilon \rightarrow 0} \limsup_{n \rightarrow \infty} \frac{1}{n} LR_{{\cal M}_n}^{\epsilon}(\rho^{\otimes n}).
\end{equation}
\end{proposition}
Taking once more $\{ {\cal M}_n \}$ as the sets of separable states 
over ${\cal H}^{\otimes n}$, Proposition \ref{relenteqrob} shows that the 
regularized relative entropy of entanglement is a smooth asymptotic 
version of the log global robustness of entanglement 
\cite{VT99, HN03, Bra05, Dat08b}. Hence we have a connection between 
the robustness of quantum correlations under mixing and their 
distinguishability to classical correlations. A different, but related, 
proof of this fact has been found in Ref. \cite{Dat08b}.

A corollary of Theorem \ref{maintheorem} is the following.
\begin{corollary} \label{faithful}
The regularized relative entropy of entanglement is faithful. For every 
entangled state $\rho \in {\cal D}({\cal H}_1 \otimes ... \otimes 
{\cal H}_n)$,
\begin{equation}
        E_R^{\infty}(\rho) > 0.
\end{equation}
\end{corollary}
Recently, Piani found an independent proof of Corollary \ref{faithful}, using completely different techniques - most notably the insight of defining a new variant of the relative entropy of entanglement, based on the optimal distinguishability of an entangled state to separable states accessible by restricted measurements, e.g. LOCC ones \cite{Piani09}. 

Corollary \ref{faithful} has an interesting consequence to the theory of asymptotic 
entanglement conversion of multipartite states. Given two states 
$\rho, \sigma \in {\cal D}({\cal H}_1 \otimes ... \otimes {\cal H}_n)$, 
we define the LOCC optimal asymptotic rate of conversion of $\rho$ into $\sigma$ as 
\begin{equation}
R(\rho \rightarrow \sigma) := \inf_{\{ k_n \}_{n \in \mathbb{N}}} \left \{ \limsup_{n \rightarrow \infty} \frac{k_n}{n}  : \lim_{n \rightarrow \infty} \left( \min_{\Lambda \in LOCC} || \Lambda(\rho^{\otimes k_n}) - \sigma^{\otimes n}||_1 \right) = 0 \right \},
\end{equation}
where the infimum is taken over all sequences of integers $\{ k_n \}_{n \in \mathbb{N}}$ and the minimization over all LOCC trace preserving maps $\Lambda$. We are therefore interested in the most efficient manner to transform a given entangled state into another, in the regime of many copies, when we only have access to LOCC. 

A fundamental question in this context is whether the rate $R(\rho \rightarrow \sigma)$ is non-zero whenever $\sigma$ is entangled. For states composed of two parties, the work of Yang \textit{et al} \cite{YHHS05} has provided the answer in the affirmative. The general case of multipartite states, however, remained open. A direct application of Corollary \ref{faithful} shows that indeed the rate function is strictly positive whenever the target state is entangled. We thus find that the mathematical definition of entanglement, as states that cannot be written as in Eq. (\ref{separable}), is equivalent to an operational definition of entangled states, as states which require a non-zero rate of entangled pure states - or any other fixed entangled state in fact - for their formation in the asymptotic limit.  

\begin{corollary} \label{positiverate}
For every two entangled states $\rho, \sigma \in {\cal D}({\cal H}_1 \otimes ... \otimes {\cal H}_n)$,
\begin{equation}
R(\rho \rightarrow \sigma) > 0.
\end{equation}
\end{corollary}

Another application of our main theorem is given in the follow up paper \cite{BPCMP09} (see also \cite{BP08b, Hor08}). There, Theorem \ref{nonlockentropyalmostppowerstates} is the key technical tool to prove reversibility in the asymptotic manipulation of entangled states under quantum operations which cannot (approximately) generate entanglement.

In the next three sections we provide the proofs of Theorem \ref{maintheorem}, Proposition \ref{relenteqrob}, Corollary \ref{faithful}, and Corollary \ref{positiverate}. 

\section{Proof of Theorem \ref{maintheorem}} \label{proofmain}

We start proving Proposition \ref{relenteqrob} and then use it to establish the following auxiliary result.

\begin{proposition} \label{maincompact} 
For every family of sets  $\{ {\cal M}_n \}_{n \in \mathbb{N}}$ satisfying properties \ref{cond1}-\ref{cond5} and every state $\rho \in {\cal D}({\cal H})$,
\begin{equation}
\lim_{n \rightarrow \infty} \min_{\omega_n \in {\cal M}_n} \tr( \rho^{\otimes n} - 2^{yn}\omega_n)_+ = 
\begin{cases}
0, & y > E_{\cal M}^{\infty}(\rho),\\
1, & y < E_{\cal M}^{\infty}(\rho).
\end{cases}
\end{equation}
\end{proposition}

Before proving Propositions \ref{relenteqrob} and \ref{maincompact}, let 
us show how Proposition \ref{maincompact} implies Theorem \ref{maintheorem}.

\vspace{0.1 cm}

\begin{proof} (Theorem \ref{maintheorem})
Consider the following family of convex optimization problems
\begin{equation}
        \lambda_n(\pi, K) := \max_{A} \left[
        \tr(A \pi) : 0 \leq A \leq 
        \id, \hspace{0.2 cm} \tr(A \sigma) \leq \frac{1}{K} 
        \hspace{0.2 cm} \forall \hspace{0.1 cm} \sigma \in {\cal M}_n
        \right].
\end{equation}
The statement of  Theorem I is immediately implied by
\begin{equation} \label{refor}
        \lim_{n \rightarrow \infty} \lambda_n(\rho^{\otimes n}, 2^{n y}) = 
\begin{cases}
        0, & y > E_{\cal M}^{\infty}(\rho), \\
        1, & y < E_{\cal M}^{\infty}(\rho).
\end{cases}
\end{equation}
In order to see that Eq. (\ref{refor}) holds true, we go to the dual formulation of $\lambda_n(\pi, K)$. We first rewrite it as 
\begin{equation}
        \lambda_n(\pi, K) := \max_{A} \left[
        \tr(A \pi) : 0 \leq A \leq 
        \id, \hspace{0.2 cm} \tr\left((\id/K - A) \sigma\right) \geq 0 
        \hspace{0.2 cm} \forall \hspace{0.1 cm} \sigma \in 
        \text{cone}({\cal M}_n)\right],
\end{equation}
where $\text{cone}({\cal M}_n)$ is the cone of ${\cal M}_n$. Then, we 
note that the second constraint is a generalized inequality (since the 
set $\text{cone}({\cal M}_n)$ is a convex proper cone) \cite{BV01} and 
write the problem as
\begin{equation} \label{primal}
        \lambda_n(\pi, K) := \max_{A} \left[
        \tr(A \pi) : 0 \leq A \leq 
        \id, \hspace{0.2 cm} (\id/K - A)  \in ({\cal M}_n)^{*}
        \right],
\end{equation}
where $({\cal M}_n)^{*}$ is the dual cone of ${\cal M}_n$. The Lagrangian 
of $\lambda_n(\pi, K)$ is given by
\begin{equation}
        L(\pi, K, A, X, Y, \mu) = \tr(A \pi) + \tr(XA) + tr(Y(\id - A)) 
        + \tr((\id/K - A)\mu), 
\end{equation}
where $X, Y \geq 0$ and $\mu \in \text{cone}({\cal M}_n)$ are Lagrange 
multipliers. It is easy to find a strictly feasible solution for the 
primal optimization problem given by Eq. (\ref{primal}) (e.g. $A = \id/(2K)$). Therefore, by 
Slater's condition \cite{BV01}, $\lambda_n(\pi, K)$ is equal to its dual 
formulation, which reads
\begin{equation}
        \lambda_n(\pi, K) = \min_{Y, \mu} \left[\tr(Y) + \tr(\mu)/K : 
        \pi \leq Y + \mu, \hspace{0.2 cm} Y \geq 0, \hspace{0.2 cm} 
        \mu \in \text{cone}({\cal M}_n)\right]. 
\end{equation}
Using that $\tr(A)_+ = \min_{Y} \tr(Y) : Y \geq 0, Y \geq A$, we find
\begin{equation}
        \lambda_n(\pi, K) = \min_{\mu} \left[\tr(\pi - \mu)_+ 
        + \tr(\mu)/K : \mu \in \text{cone}({\cal M}_n)\right], 
\end{equation}
which can finally be rewritten as
\begin{equation}
        \lambda_n(\pi, K) = \min_{\mu, b} \left[ \tr(\pi - b\mu)_+ + b/K : 
        \mu \in {\cal M}_n, \hspace{0.2 cm} b \in \mathbb{R}_+\right]. 
\end{equation}

Let us consider the asymptotic behavior of $\lambda_n(\rho^{\otimes n}, 
2^{n y})$. Take $y = E_{\cal M}^{\infty}(\rho) + \epsilon$, for any 
$\epsilon > 0$. Then we can choose $b = 2^{n(E_{\cal M}^{\infty}(\rho) + \frac{\epsilon}{2})}$, giving
\begin{equation} 
        \lambda_n(\rho^{\otimes n}, 2^{ny}) \leq \min_{\mu \in {\cal M}_n} 
        \left[\tr(\rho^{\otimes n} - 2^{n(E_{\cal M}^{\infty}(\rho) + 
        \frac{\epsilon}{2})}\mu)_+ + 2^{-n\frac{\epsilon}{2}}\right].
\end{equation}
From Proposition \ref{maincompact} we then find that 
$\lambda_n(\rho^{\otimes n}, 2^{ny}) \rightarrow 0$. 

We now take $y = E_{\cal M}^{\infty}(\rho) - \epsilon$, for any 
$\epsilon > 0$. The optimal $b$ for each $n$ has to satisfy $b_n 
\leq 2^{y n}$, otherwise $\lambda_n(\rho^{\otimes n}, 2^{ny})$ would 
be larger than one, which we know is false. Therefore,
\begin{equation} 
        \lambda_n(\rho^{\otimes n}, 2^{ny}) \geq \min_{\mu \in {\cal M}_n} 
        \tr(\rho^{\otimes n} - 2^{n(E_{\cal M}^{\infty}(\rho) - \epsilon)}
        \mu)_+,
\end{equation}
which approaches unity again by Proposition \ref{maincompact}.
\end{proof}

\subsection{Proof of Proposition \ref{relenteqrob}}

\begin{proof} (Proposition \ref{relenteqrob})

We start showing that 
\begin{equation}
        E_{\cal M}^{\infty}(\rho) \leq \lim_{\epsilon \rightarrow 0} 
        \liminf_{n \rightarrow \infty} \frac{1}{n} 
        LR_{{\cal M}_n}^{\epsilon}(\rho^{\otimes n}).
\end{equation} 
Let $\rho_n^{\epsilon} \in B_{\epsilon}(\rho^{\otimes n})$ be an optimal 
state for $\rho^{\otimes n}$ in Eq. (\ref{smooth}). For every $n$ there 
is a state  $\sigma_n \in {\cal M}_n$ such that $\rho_n^{\epsilon} 
\leq s_n \sigma_n$, with $LR_{{\cal M}_n}^{\epsilon}(\rho^{\otimes n}) 
= LR_{{\cal M}_n}(\rho_n^{\epsilon}) = \log(s_n)$. It follows from the operator monotonicity of the $\log$ function \cite{Bat96} that 
if $\rho \leq 2^k \sigma$ (where $\rho$ and $\sigma$ are two states), 
then $S(\rho || \sigma) \leq k$. Hence, 
\begin{equation}
        \frac{1}{n}E_{{\cal M}_n}(\rho_n^{\epsilon}) \le \frac{1}{n}
        S(\rho_n^{\epsilon} || \sigma_n) \leq \frac{1}{n} \log s_n
        = \frac{1}{n}
        LR_{{\cal M}_n}(\rho_n^{\epsilon}) = \frac{1}{n}LR_{{\cal M}_n}^{\epsilon}(\rho^{\otimes n}).
        \label{1456}
\end{equation}

As $\rho_n^{\epsilon} \in B_{\epsilon}(\rho^{\otimes n})$, we find 
from Lemma \ref{asympcont} (see appendix \ref{usefulresult}) that
\begin{equation}
        \frac{1}{n}E_{{\cal M}_n}(\rho^{\otimes n}) \leq 
        \frac{1}{n}LR_{{\cal M}_n}^{\epsilon}(\rho^{\otimes n}) + 
        f(\epsilon),
        \label{1456a}
\end{equation}
where $f: \mathbb{R} \rightarrow \mathbb{R}$ is such that 
$\lim_{\epsilon \rightarrow 0} f(\epsilon) = 0$. Taking the limits 
$n \rightarrow \infty$ and $\epsilon \rightarrow 0$ in both sides 
of the equation above,
\begin{equation}
        E^{\infty}_{{\cal M}}(\rho) = 
        \liminf_{n \rightarrow \infty}  \frac{1}{n}E_{{\cal M}_n}
        (\rho^{\otimes n})\leq \lim_{\epsilon \rightarrow 0} 
        \liminf_{n \rightarrow \infty} \frac{1}{n}LR_{{\cal M}_n}^{\epsilon}(\rho^{\otimes n}).
\end{equation}

To show the converse inequality, namely that 
\begin{equation}
        E_{\cal M}^{\infty}(\rho) \geq \lim_{\epsilon \rightarrow 0} 
        \limsup_{n \rightarrow \infty} \frac{1}{n} 
        LR_{{\cal M}_n}^{\epsilon}(\rho^{\otimes n}),
\end{equation} 
let $y_k := E_{{\cal M}_k}(\rho^{\otimes k})
+ \varepsilon = S(\rho^{\otimes k} || \sigma_k) + \varepsilon$ 
($\sigma_k$ is an optimal state for $\rho^{\otimes k}$ in 
$E_{{\cal M}_k}(\rho^{\otimes k})$) with $\varepsilon > 0$. We can 
write for every $n \in \mathbb{N}$,
\begin{equation}\label{inehj}
        \rho^{\otimes kn} \leq 2^{y_k n} \sigma_k^{\otimes n} + 
        (\rho^{\otimes kn} - 2^{y_k n} \sigma_k^{\otimes n})_+.
\end{equation}
From Lemma \ref{ON} (see appendix \ref{usefulresult}) we have

\begin{equation}
        \lim_{n \rightarrow \infty}tr(\rho^{\otimes kn} - 2^{y_k n} 
        \sigma_k^{\otimes n})_+ = 0.
\end{equation}

Applying Lemma \ref{DR} (see appendix \ref{usefulresult}) to 
Eq. (\ref{inehj}) we then find that there is a sequence of states 
$\rho_{n, k}$ such that 
\begin{equation}
        \lim_{n \rightarrow \infty} || \rho^{\otimes kn} - \rho_{n, k} ||_1 = 0
\end{equation}
and
\begin{equation}
        \rho_{n, k} \leq g(n)2^{y_k n} \sigma_k^{\otimes n},
\end{equation}
where $g: \mathbb{R}_+ \rightarrow \mathbb{R}_+$ is such that 
$\lim_{n \rightarrow \infty}g(n) = 1$. It follows that for every 
$\delta > 0$ there is a sufficiently large $n_0$ such that for all 
$n \geq n_0$, $\rho_{n, k} \in B_{\delta}(\rho^{\otimes kn})$. 
Moreover, from property \ref{cond4} of the sets we find 
$\sigma_k^{\otimes n} \in {\cal M}_{kn}$. Hence, for every $\delta > 0$,
\begin{equation} \label{eqinthemiddle}
        \limsup_{n \rightarrow \infty} 
        \frac{LR_{{\cal M}_{nk}}^{\delta}(\rho^{\otimes nk}) }{n}
        \leq 
        \limsup_{n \rightarrow \infty} \frac{LR_{{\cal M}_{kn}}(\rho_{n, k})}{n} 
        \leq y_k = E_{{\cal M}_k}(\rho^{\otimes k}) + \varepsilon.
\end{equation}

The next step is to note that for every $k \in \mathbb{N}$,
\begin{equation} \label{foot1}
\limsup_{n \rightarrow \infty} \frac{1}{nk}LR_{{\cal M}_{nk}}^{\delta}(\rho^{\otimes nk}) = \limsup_{n \rightarrow \infty} \frac{1}{n}LR_{{\cal M}_{n}}^{\delta}(\rho^{\otimes n}).
\end{equation} 
The $\leq$ inequality follows straightforwardly. For the $\geq$ inequality, let $\{ n' \}$ be a subsequence such that 
\begin{equation}
M := \lim_{n' \rightarrow \infty} \frac{1}{n'}LR_{{\cal M}_{n'}}^{\delta}(\rho^{\otimes n'})
\end{equation}
is equal to the R.H.S. of Eq. (\ref{foot1}). Let $n_k'$ be the first multiple of $k$ larger than $n'$. Then, 
\begin{eqnarray}
\limsup_{n \rightarrow \infty} \frac{1}{nk}LR_{{\cal M}_{nk}}^{\delta}(\rho^{\otimes nk}) &\geq& \limsup_{n_k' \rightarrow \infty} \frac{1}{n_k'} LR_{{\cal M}_{n_k'}}^{\delta}(\rho^{\otimes n_k'}) \nonumber \\ &\geq& \limsup_{n_k' \rightarrow \infty} \frac{1}{n_k'} LR_{{\cal M}_{n'}}^{\delta}(\rho^{\otimes n'}) \nonumber \\ &=& M.  
\end{eqnarray}
The last inequality follows from $LR_{{\cal M}_{n}}^{\delta}(\pi) \geq LR_{{\cal M}_{n-l}}^{\delta}(\tr_{1,..l}(\pi))$, which is a consequence of property \ref{cond3} of the sets.

From Eq. (\ref{eqinthemiddle}) and the fact that $\varepsilon, \delta > 0$ are arbitrary, it follows that
\begin{equation}
\lim_{\delta \rightarrow 0} \limsup_{n \rightarrow \infty} \frac{1}{n} LR_{{\cal M}_{n}}^{\delta}(\rho^{\otimes n}) \leq \frac{1}{k}E_{{\cal M}_k}(\rho^{\otimes k}).
\end{equation}
Finally, since the above equation is true for every $k \in \mathbb{N}$, we find the announced result.
\end{proof}
\vspace{0.3 cm}

There is another related quantity that we might consider in this context, in which $\epsilon$ and $n$ are not independent. Define
\begin{equation}
LG_{{\cal M}}(\rho) := \inf_{\{ \epsilon_n\}} \left \{ \limsup_{n \rightarrow \infty} \frac{1}{n} LR_{{\cal M}_n}^{\epsilon_n}(\rho^{\otimes n}) : \lim_{n \rightarrow \infty} \epsilon_n = 0 \right \}. 
\end{equation}
The proof of Proposition \ref{relenteqrob} can be straightforwardly adapted to show
\begin{corollary} \label{LGmeasure}
For every family of sets  $\{ {\cal M}_n \}_{n \in \mathbb{N}}$ satisfying properties \ref{cond1}-\ref{cond5} and every quantum state $\rho \in {\cal D}({\cal H})$,
\begin{equation}
LG_{{\cal M}}(\rho) = E_{\cal M}^{\infty}(\rho).
\end{equation}
\end{corollary}

With Proposition \ref{relenteqrob} at hand we are now in position to prove the strong converse part of Proposition \ref{maincompact}, which we restate as a separate corollary for the sake of clarity.
\begin{corollary} \label{maincompactstrongconverse}
Let $\rho \in {\cal D}({\cal H})$. For every $y > E_{R}^{\infty}(\rho)$
\begin{equation}
\lim_{n \rightarrow \infty} \min_{\omega_n \in {\cal M}_n} \tr(\rho^{\otimes n} - 2^{y n}\omega_n)_+ = 0,
\end{equation}
while for every $y < E_{\cal M}^{\infty}(\rho)$,
\begin{equation}
\liminf_{n \rightarrow \infty} \min_{\omega_n \in {\cal M}_n} \tr(\rho^{\otimes n} - 2^{y n}\omega_n)_+ > 0,
\end{equation}
\end{corollary}
\begin{proof} 

We first show that if $y = E_{\cal M}^{\infty}(\rho) + 
\epsilon$, then
\begin{equation}\label{onepart}
        \lim_{n \rightarrow \infty} \min_{\omega_n \in {\cal M}_n} 
        \tr( \rho^{\otimes n} - 2^{yn}\omega_n)_+ = 0.
\end{equation}
By Proposition \ref{relenteqrob} there is a $\delta_0 > 0$ such that
\begin{equation} \label{dif1}
        \left | E_{\cal M}^{\infty}(\rho)  - \limsup_{n \rightarrow \infty} 
        \frac{1}{n} LR_{{\cal M}_n}^{\delta}(\rho^{\otimes n}) \right | 
        \leq \epsilon/2,
\end{equation}
for every $\delta \leq \delta_0$. Let $\rho_{n, \delta} \in B_{\delta}
(\rho^{\otimes n})$ be an optimal state in Eq. (\ref{smooth}) for 
$\rho^{\otimes n}$  realizing the value 
$LR^{\delta}_{{\cal M}_n}(\rho^{\otimes n})$. Then there must exist a 
$\sigma_n \in {\cal M}_n$ such that 
\begin{equation}
        \rho_{n, \delta} \leq 2^{LR_{{\cal M}_n}^{\delta}(\rho^{\otimes n})} 
        \sigma_n,
\end{equation}
from which follows that for every $\lambda \geq LR_{{\cal M}_n}^{\delta}
(\rho^{\otimes n})/n$,
\begin{equation}
        \min_{\omega_n \in {\cal M}_n} \tr( \rho^{\otimes n} - 
        2^{\lambda n}\omega_n)_+ \leq \min_{\omega_n \in {\cal M}_n} 
        \tr( \rho_{n, \delta} - 2^{{\lambda} n}\omega_n)_+ + \delta 
        \leq \delta.
\end{equation}
 
From Eq. (\ref{dif1}) and our choice of $y$ we then find that for every 
$\delta > 0$ there is a sufficiently large $n_0$ such that for all 
$n \geq n_0$,
\begin{equation}
        \min_{\omega_n \in {\cal M}_n} \tr( \rho^{\otimes n} - 
        2^{y n} \omega_n )_+ \leq \delta,
\end{equation}
from which Eq. (\ref{onepart}) follows.  

Now we move to the second part of the proof which aims to show that
that if $y = E_{\cal M}^{\infty}(\rho) - \epsilon$, then 
\begin{equation}\label{otherpart}
        \liminf_{n \rightarrow \infty} \min_{\omega_n \in {\cal M}_n} 
        \tr( \rho^{\otimes n} - 2^{yn}\omega_n)_+ > 0.
\end{equation}
To this end, let us assume by means of a contradiction that this is not the case and that the limit 
is zero. For each $n$ we have
\begin{equation} \label{ine11}
        \rho^{\otimes n} \leq 2^{yn}\omega_n + (\rho^{\otimes n} - 
        2^{yn}\omega_n)_+,
\end{equation}
where $\omega_n$ is the optimal state  in ${\cal M}_n$ in Eq. (\ref{otherpart}). Applying Lemma \ref{DR} to Eq. (\ref{ine11}) we then find that there is a 
sequence of states $\tilde{\rho}_n$ (for an increasing subsequence ${\cal F} \subseteq \mathbb{N}$, $\{ n \}_{n \in {\cal F}}$ such that $||\rho^{\otimes n} - \tilde{\rho}_{n}  
||_1 \rightarrow 0$ and $\tilde{\rho}_{n} \leq g(n)2^{yn}\omega_{n}$, for 
a function $g$ satisfying $\lim_{n\rightarrow \infty}g(n) = 1$. It 
follows that 
\begin{equation}
        \frac{1}{n}LR_{{\cal M}_{n}}(\tilde{\rho}_{n}) \leq y + \frac{\log g(n)}{n}
\end{equation}
and that for every $\delta > 0$ and sufficiently large $n$, 
$\tilde{\rho}_{n} \in B_{\delta}(\rho^{\otimes n})$. Therefore, for 
every $\delta > 0$,
\begin{equation}
        \liminf_{n \rightarrow \infty} \frac{1}{n}LR_{{\cal M}_n}^{\delta}
        (\rho^{\otimes n}) \leq \liminf_{n \rightarrow \infty} 
        \frac{1}{n}LR_{{\cal M}_n}(\tilde{\rho}_n) \leq y = 
        E_{{\cal M}}^{\infty}(\rho) - \epsilon,
\end{equation}
in contradiction to Eq. (\ref{relenteqrobeq}) of Proposition 
\ref{relenteqrob}.
\end{proof}

\subsection{Proof of the direct part of Proposition \ref{maincompact}}

We now turn to the proof of the direct part of Proposition \ref{maincompact}, which is 
the main technical contribution of the paper. Before we start with 
the proof in earnest, we provide a rough outline of the main steps 
which will be taken, in order to make the presentation more transparent. 

In Corollary \ref{maincompactstrongconverse} we showed by relatively simple means that $E_{\cal M}^{\infty}(\rho)$ is the strong converse rate for the hypothesis testing problem which we are analysing. It is more involved to show that $E_{\cal M}^{\infty}(\rho)$ is also 
an achievable rate, i.e. that the limit equals unity for every $y < 
E_{\cal M}^{\infty}(\rho)$. The difficulty is precisely that the 
alternative hypothesis is non-i.i.d. and is a set of states, instead of a single one in general. Most of the proof is devoted to 
circumvent this problem. The main ingredient of the proof is a variant of 
Renner's exponential version of the quantum de Finetti theorem 
\cite{Ren05, Ren07} (see Appendix \ref{Rennerexp}), given in Lemma 
\ref{variantexpdefineti}. 

Loosely speaking, we will proceed as follows. We will show the reverse implication that if 
\begin{equation} \label{ouline}
\min_{\omega_n \in {\cal M}_n} \tr( \rho^{\otimes n} - 2^{y n} \omega_n)_+ \stackrel{n \rightarrow \infty}{\longrightarrow} \mu < 1
\end{equation}
then $y \geq E_{\cal M}^{\infty}(\rho) - o(1)$. To this aim we first use 
Lemma \ref{DR} (see appendix \ref{usefulresult}) to find from the equation above a state 
$\rho_n$ that possesses non-negligible fidelity with $\rho^{\otimes n}$ and satisfies
\begin{equation}
        \rho_n \leq 2^{yn + o(n)}\omega_n,
\end{equation}
for every $n$, where $\omega_n \in {\cal M}_n$ is the optimal state 
in the minimization of Eq. (\ref{ouline}). Due to property \ref{cond5} 
of the sets, we can take $\omega_n$ and thus also $\rho_n$ to be 
permutation-symmetric. Then, tracing a sublinear number of copies $o(n)$ 
and using Lemmata \ref{ulhmansymmetric} and \ref{variantexpdefineti} we will be able to show that the previous equation 
implies that there is a state $\pi_{\rho, n}$ exponential close to an almost power state along $\rho$ (see Eq. (\ref{almostiid}) for a definition) such that 
\begin{equation} \label{eqeqffdhgd}
       \pi_{\rho, n} \leq 2^{yn + o(n)}\tr_{1,...,o(n)}(\omega_n).
\end{equation}

In a second part of the proof, we will argue that the measure $E_{{\cal M}_n}(\pi_{\rho, n})$ is not too far away from $E_{{\cal M}_n}(\rho^{\otimes n})$, with the difference being upper bounded by a term sublinear in $n$. This property can be considered as a manifestation of the non-lockability of the measures $E_{{\cal M}_n}$, as was proved 
for the relative entropy of entanglement in Ref. \cite{HHHO05}. 

Finally, using the operator monotonicity of the $\log$ and the 
asymptotic continuity of both $E_{{\cal M}_k}$ and $E_{\cal M}^{\infty}$ 
(see Appendix \ref{usefulresult}), we will find from Eq. (\ref{eqeqffdhgd}) that, for sufficiently large $n$,
\begin{equation}
       E_{\cal M}^{\infty}(\rho) = \frac{1}{n}E_{{\cal M}_{n - o(n)}}(\pi_{\rho, n}) + o(1) \leq y + o(1).
\end{equation}

The next lemma is an extension of Uhlmann's theorem on the fidelity 
\cite{Uhl76} to the case of tensor product and symmetric states.

\begin{lemma} \label{ulhmansymmetric}
Let $\rho \in {\cal D}({\cal H})$ and $\rho_n \in {\cal D}({\cal H}^{\otimes n})$ 
be such that $\hat{S}_n(\rho_n) = \rho_n$. Then there is a purification 
$\ket{\theta} \in {\cal H}\otimes {\cal H}$ of $\rho$ and a permutation-symmetric 
purification $\ket{\Psi_n} \in ({\cal H} \otimes {\cal H})^{\otimes n}$ of $\rho_n$ 
such that $|\braket{\Psi_n}{\theta^{\otimes n}}| = F(\rho_n, \rho^{\otimes n})$. 
\end{lemma}
\begin{proof}
Let $\ket{\phi^+} := \sum_{k=1}^{\dim({\cal H})}\ket{k, k}$ and consider the following purifications of $\rho$ and $\rho_n$, respectively: $\ket{\theta} = \id \otimes \sqrt{\rho}\ket{\phi^+}$ and $\ket{\Psi_n} = \id^{\otimes n} \otimes (\sqrt{\rho_n}U)\ket{\phi^+}^{\otimes n}$, where the unitary $U$ is a particular unitary, to be specified in the next paragraph, such that $\sqrt{\rho_n} \sqrt{\rho^{\otimes n}} = U|\sqrt{\rho_n}\sqrt{\rho^{\otimes n}}|$ \cite{Bat96}. A direct calculation shows that $|\braket{\Psi_n}{\theta^{\otimes n}}| = F(\rho_n, \rho^{\otimes n})$. 

To see that $\ket{\Psi_n}$ is permutation-symmetric, we note that as $\rho^{\otimes n}$ and $\rho_n$ are permutation-invariant, we can take $U$ and thus $\sqrt{\rho_n}U$ to be invariant under permutations too. Indeed, as $\sqrt{\rho_n} \sqrt{\rho^{\otimes n}}$ and $|\sqrt{\rho_n}\sqrt{\rho^{\otimes n}}|$ are permutation invariant, we can write them in the Schur basis \cite{FH91} as
\begin{equation}
\sqrt{\rho_n} \sqrt{\rho^{\otimes n}} = \bigoplus_{\lambda} A_{\lambda} \otimes \id_{\lambda}, \hspace{0.2 cm} |\sqrt{\rho_n}\sqrt{\rho^{\otimes n}}| = \bigoplus_{\lambda} B_{\lambda} \otimes \id_{\lambda},
\end{equation}
where $\lambda$ labels the irreps of $S_n$, $\id_{\lambda}$ is the identity on the irrep labelled by $\lambda$, and $A_{\lambda}, B_{\lambda}$ are operators acting on the multiplicity space of the the irrep labelled by $\lambda$ \cite{FH91}. We can define the partial isometry $V$ as  
\begin{equation}
V := \sqrt{\rho_n} \sqrt{\rho^{\otimes n}} |\sqrt{\rho_n}\sqrt{\rho^{\otimes n}}|^{-1} = \bigoplus_{\lambda} A_{\lambda}B_{\lambda}^{-1} \otimes \id_{\lambda}, 
\end{equation}
where the inverses are taken in the generalized sense. As each $A_{\lambda}B_{\lambda}^{-1}$ is a partial isometry, we can extend them to unitaries $U_{\lambda}$. Then we set 
\begin{equation}
U := \bigoplus_{\lambda} U_{\lambda} \otimes \id_{\lambda}, 
\end{equation}
which is clearly permutation-invariant. 
 
Finally, for every permutation $\pi \in S_n$,
\begin{equation}
P_{\pi}\ket{\Psi_n} = P_{\pi, S} \otimes P_{\pi, E} (\id \otimes \sqrt{\rho_n}U)\ket{\phi^+}^{\otimes n} = \id \otimes (P_{\pi, E}\sqrt{\rho_n}U P_{\pi, E})(P_{\pi, S}\otimes P_{\pi, E})\ket{\phi^+}^{\otimes n} = \ket{\Psi_n}.
\end{equation}
\end{proof}

\vspace{0.8 cm}

The next lemma can be seen as a post-selected variant of the exponential de Finetti theorem \cite{Ren05, Ren07} and is proved by similar techniques. For a $\ket{\theta} \in {\cal H}$ and $0 \leq r \leq n$ we define the set of $\binom{n}{r}$-i.i.d 
states in $\ket{\theta}$ as 
\begin{equation}
        {\cal V}({\cal H}^{\otimes n}, \ket{\theta}^{\otimes n - r}) := 
        \{ P_{\pi}( \ket{\theta}^{\otimes n - r}
        \otimes \ket{\psi_r}) : 
        \pi \in S_n, \ket{\psi_r} \in {\cal H}^{\otimes r}   \}.
\end{equation}
Thus for every state in ${\cal V}({\cal H}^{\otimes n}, \ket{\theta}^{\otimes n - r})$ we have the state $\ket{\theta}$ in at 
least $n - r$ of the copies. The set of almost power states in 
$\ket{\theta}$ is defined as \cite{HHH+08a, HHH+08b}
\begin{equation} \label{almostiid}
        \ket{\theta}^{[\otimes, n, r]} := \text{Sym}({\cal H}^{\otimes n}) 
        \cap \text{span} ({\cal V}({\cal H}^{\otimes n}, 
        \ket{\theta}^{\otimes n - r})). 
\end{equation}
Finally, we say a mixed state $\rho_n \in {\cal D}({\cal H}^{\otimes n})$ is an almost power state along $\sigma \in {\cal D}({\cal H})$, if there is a purification of $\rho_n$, $\ket{\psi} \in {\cal H}^{\otimes n} \otimes {\cal H}_E^{\otimes n}$, where ${\cal H}_E \cong {\cal H}$ is the purifying Hilbert space, such that $\ket{\psi} \in \ket{\theta}^{[\otimes, n, r]}$, for some purification $\ket{\theta} \in {\cal H}\otimes {\cal H}_E$ of $\sigma$.

\begin{lemma} \label{variantexpdefineti}
Let $\ket{\Psi_n} \in {\cal H}^{\otimes n}$ be a permutation-invariant state 
and $\ket{\theta} \in {\cal H}$. Then for every $m \leq n$ there is a state 
$\ket{\Psi_{n, m}} \in {\cal H}^{\otimes n - m}$ such that 
\begin{equation}
        \ket{\Psi_{n, m}}\bra{\Psi_{n, m}} \leq 
        |\braket{\Psi_{n}}{\theta^{\otimes n}}|^{-2} 
        \tr_{1, ..., m}(\ket{\Psi_n}\bra{\Psi_n}),
\end{equation}
and for every $r \leq n - m$ 
\begin{equation}
        || \ket{\Psi_{n, m}}\bra{\Psi_{n, m}} - \ket{\Psi_{n, m, r}}
        \bra{\Psi_{n, m, r}} ||_1 \leq 2\sqrt{2} |\braket{\Psi_n}{\theta^{\otimes n}}|^{-1} 
         e^{-\frac{mr}{2n}}
\end{equation} 
for an almost power state $\ket{\Psi_{n, m, r}} \in 
\ket{\theta}^{[\otimes, n - m, r]}$.
\end{lemma}

\begin{proof}
We write $\ket{\Psi_n} = \braket{\theta^{\otimes n}}{\Psi_n}
\ket{\theta}^{\otimes n} 
+ \sqrt{1 - |\braket{\theta^{\otimes n}}{\Psi_n}|^2}\ket{\Phi_n}$, where $\ket{\Phi_n}$ 
is a permutation-symmetric state orthogonal to $\ket{\theta}^{\otimes n}$. We can 
expand $\ket{\Phi_n}$ as $\ket{\Phi_n} = \sum_{k=1}^{n} \beta_k 
\text{Sym}(\ket{\eta_k} \otimes \ket{\theta}^{\otimes n - k})$, where $\ket{\eta_k}$ are permutation-symmetric states which live in $({\cal H} \bot \ket{\theta})^{\otimes k}$ and $\sum_k |\beta_k|^2 = 1$. 

Define $\ket{\Psi_{n, m}} := (\bra{\theta}^{\otimes m} \otimes \id^{\otimes n - m})\ket{\Psi_n}/||  (\bra{\theta}^{\otimes m} \otimes \id^{\otimes n - m})\ket{\Psi_n} ||$. From the inequality 
\begin{equation} \label{label143}
||  (\bra{\theta}^{\otimes m} \otimes \id^{\otimes n - m})\ket{\Psi_n} || := \bra{\Psi_n}(\ket{\theta}\bra{\theta})^{\otimes m} \otimes \id^{\otimes n - m}\ket{\Psi_n}^{1/2} \geq |\braket{\Psi_n}{\theta^{\otimes n}}|
\end{equation}
we find
\begin{eqnarray}
\ket{\Psi_{n, m}}\bra{\Psi_{n, m}} &\leq& ||  (\bra{\theta}^{\otimes m} \otimes \id^{\otimes n - m})\ket{\Psi_n} ||^{-2} \tr_{1, ..., m}(\ket{\Psi_n}\bra{\Psi_n}) \nonumber \\ &\leq& |\braket{\Psi_n}{\theta^{\otimes n}}|^{-2} \tr_{1, ..., m}(\ket{\Psi_n}\bra{\Psi_n}).  
\end{eqnarray}

To estimate how close $\ket{\Psi_{n, m}}$ is to an almost power state, we make use of the following relation, valid for every $m \leq n$, 
\begin{equation} \label{releasy1}
(\bra{\theta}^{\otimes m} \otimes \id^{\otimes n - m})\text{Sym}(\ket{\eta_k} \otimes \ket{\theta}^{\otimes n - k}) = \binom{n}{k}^{-1/2}\binom{n - m}{k}^{1/2} \text{Sym}(\ket{\eta_k} \otimes \ket{\theta}^{\otimes n - k - m}). 
\end{equation}
Define
\begin{eqnarray}
\ket{\Psi'_{n, m, r}} &:=& ||  (\bra{\theta}^{\otimes m} \otimes 
\id^{\otimes n - m})\ket{\Psi_n} ||^{-1} 
( \braket{\Psi_n}{\theta^{\otimes n}}\ket{\theta}^{\otimes n - m}  \nonumber \\ &+& \sqrt{1 - |\braket{\Psi_n}{\theta^{\otimes n}}|^2} \sum_{k=1}^{r} \beta_k \binom{n}{k}^{-1/2}\binom{n - m}{k}^{1/2} \text{Sym}(\ket{\eta_k} \otimes \ket{\theta}^{\otimes n - k - m}) ).
\end{eqnarray}
Note that $\ket{\Psi'_{n, m, n}} = \ket{\Psi_{n, m}}$. Then, from Eq. (\ref{label143}), 
\begin{eqnarray}
|| \ket{\Psi'_{n, m, r}} - \ket{\Psi_{n, m}} || &\leq& |\braket{\Psi_n}{\theta^{\otimes n}}|^{-1} \Vert \sum_{k=r+1}^{n} \beta_k \binom{n}{k}^{-1/2}\binom{n - m}{k}^{1/2} \text{Sym}(\ket{\eta_k} \otimes \ket{\theta}^{\otimes n - k - m}) \Vert \nonumber \\ &=& |\braket{\Psi_n}{\theta^{\otimes n}}|^{-1} \left( \sum_{k=r+1}^n |\beta_k|^2 \binom{n}{k}^{-1}\binom{n - m}{k} \right)^{\frac{1}{2}}.
\end{eqnarray}
We have
\begin{eqnarray}
        \binom{n}{k}^{-1}\binom{n - m}{k} &=& 
        \frac{(n-m)(n - m - 1)...(n - m - k + 1)}{n (n - 1)...(n - k + 1)} \nonumber \\ 
        &=& \left(1 - \frac{m}{n}\right) ... \left( 1 - \frac{m}{n - k + 1} \right)\nonumber\\ 
        &\leq& \left(1 - \frac{m}{n}\right)^{k} \leq e^{-\frac{mk}{n}}.
\end{eqnarray}
where we used that for $\beta \in (0, 1]$, $(1 - \beta)^{1/\beta} \leq e^{-1}$. Hence
\begin{equation} \label{primeeqq}
        || \ket{\Psi'_{n, m, r}} - \ket{\Psi_{n, m}} || \leq 
        |\braket{\Psi_n}{\theta^{\otimes n}}|^{-1} 
        \left(\sum_{k=r+1}^{n} e^{-\frac{mk}{n}}|\beta_k|^2\right)^{\frac{1}{2}} 
        \leq |\braket{\Psi_n}{\theta^{\otimes n}}|^{-1} e^{-\frac{mr}{2n}},
\end{equation}
where in the last inequality we used that $\sum_{k=r+1}^n |\beta|_k^2 \leq 1$.

Defining $\ket{\Psi_{n, m, r}} := \ket{\Psi'_{n, m, r}}/||\ket{\Psi'_{n, m, r}}||$, we 
have $|| \ket{\Psi_{n, m, r}} - \ket{\Psi_{n, m}} || \leq 2|| \ket{\Psi'_{n, m, r}} - \ket{\Psi_{n, m}} || \leq 2 
|\braket{\Psi_n}{\theta^{\otimes n}}|^{-1} e^{-\frac{mr}{2n}}$,
where we used the estimate 
\begin{equation}
\left \Vert \frac{x}{\Vert x \Vert} -  y \right \Vert \leq \Vert x - y \Vert + \Vert \frac{x}{\Vert x \Vert} - x \Vert = \Vert x - y \Vert + 1 - \Vert x \Vert = \Vert x - y \Vert + \Vert y \Vert - \Vert x \Vert \leq 2 \Vert x - y \Vert,
\end{equation}
with $x := \ket{\Psi_{n, m, r}'}$ and $y := \ket{\Psi_{n, m}}$.

The lemma is now a consequence of the inequality $|| \ket{\psi}\bra{\psi} - \ket{\phi}\bra{\phi} ||_1 \leq \sqrt{\braket{\psi}{\psi} + \braket{\phi}{\phi}} ||\ket{\psi} - \ket{\phi} ||$ (see e.g. Lemma A.2.5 of \cite{Ren05}). 
\end{proof}

\vspace{0.8 cm}

The next lemma is an analogue of a result of Ogawa and Nagaoka \cite{ON00}, stated in Appendix \ref{usefulresult} as Lemma \ref{ON}, and originally used to establish the strong converse of quantum Stein's lemma.

\begin{lemma} \label{ogawanagaokaadaptation}
Given two states $\rho, \sigma \in {\cal D}({\cal H})$ such that $\text{supp}(\rho) \subseteq \text{supp}(\sigma)$ and real numbers $\lambda, \mu$, 
\begin{equation} \label{eqhayashicompsd}
\tr(\rho^{\otimes n} - 2^{\lambda n}\sigma^{\otimes n})_+ \leq 2^{-n(s\mu - \log \tr(\rho^{1 + s}))} + 2^{-n(s(\lambda - \mu) - s \dim({\cal H})\frac{\log(1 + n)}{n} - \log \tr(\rho \sigma^{-s}))}. 
\end{equation}
for every $s \in [0, 1]$.
\end{lemma}
\begin{proof}
Let $Q_n$ be the projector onto the positive part of $(\rho^{\otimes n} - 2^{\lambda n}\sigma^{\otimes n})$. Let $Q_n = \sum_{i} \lambda_i E_i$ be an eigen-decomposition of $Q_n$ with eigenvalues $\lambda_i$ (either equal to $0$ or $1$) and eigen-projectors $\{ E_i \}$ whose particular form will be specified later on in the proof.

Define the probability distributions $p_n(i) := \tr(\rho^{\otimes n}E_i)$ and $q_n(i) := \tr(\sigma^{\otimes n}E_i)$. From Lemma \ref{hancover} we can write  
\begin{eqnarray} \label{liuwqhfn}
\tr(\rho^{\otimes n} - 2^{\lambda n}\sigma^{\otimes n})_+ &=& \sum_{i}\lambda_i
\left(p_n(i) - 2^{\lambda n}q_n(i)\right)  \\
&\leq& \Pr_{\{p_n\}}\left( i : \frac{1}{n} \log \frac{p_n(i)}{q_n(i)} > \lambda \right) \nonumber \\ & \leq &  \Pr_{\{p_n\}}\left( i : \frac{1}{n} \log p_n(i) \geq \mu \right) \nonumber + \Pr_{\{p_n\}}\left( i : - \frac{1}{n}\log q_n(i) \geq \lambda - \mu \right)
\end{eqnarray}
for every $\mu \in \mathbb{R}$. Given a discrete probability distribution $r$, a random variable $X$, and a real number $a$, Cram\'er Theorem gives \cite{DZ98}
\begin{equation}
        \Pr_{\{r\}}(X \geq a) \leq 2^{-\Lambda(X, r, a)}, \hspace{0.3 cm} 
        \Lambda(X, r, a) := \sup_{0 \leq s \leq 1} \left(as - \log \sum_i r(i) 2^{sX(i)}\right)
\end{equation}
Applying it to the two last terms of Eq. (\ref{liuwqhfn}), 
\begin{eqnarray} \label{labelpsdjkdvjdsf}
        - \log \left( \Pr_{\{p_n\}}\left( i : \frac{1}{n} \log p_n(i) \geq \mu \right)\right) 
        &\geq& \sup_{0\leq s \leq 1} \left(s n \mu - \log
        \sum_i p_n(i)^{1 + s}\right), \nonumber\\ 
        - \log \left( \Pr_{\{p_n\}}\left( i : - \frac{1}{n}\log q_n(i) \geq \lambda - \mu \right)
        \right) &\geq& \sup_{0\leq s \leq1} \left(s n (\lambda - \mu) - 
        \log \sum_i p_n(i)q_n(i)^{-s}\right).
\end{eqnarray}

From the joint convexity of $\tr(A^{s}B^{1-s})$ for $-1 < s < 0$ \cite{Ando, JR09} we find that the function $g_s(\rho, \sigma) := \tr(\rho^{1+s}\sigma^{-s})$ is monotonic decreasing under trace preserving CP maps for every $0 < s < 1$. Defining the quantum operation ${\cal E}(X) = \sum_i E_i X E_i$, 
\begin{eqnarray} \label{bbound111}
        \sum_i p_n(i)^{1 + s} &=& \dim({\cal H})^{-ns} g_s\left({\cal E}(\rho^{\otimes n}), {\cal E}\left(\frac{\id^{\otimes n}}{\dim({\cal H})^n}\right)\right) \nonumber \\ &\leq& \dim({\cal H})^{-ns} g_s\left(\rho^{\otimes n}, \frac{\id^{\otimes n}}{\dim({\cal H})^n}\right) = \tr((\rho^{\otimes n})^{1 + s}). 
\end{eqnarray}
Applying it to the first inequality in Eq. (\ref{labelpsdjkdvjdsf}) gives the first 
term on the right hand side in Eq. (\ref{eqhayashicompsd}).

For the second bound, we first note that the permutation-invariance of $R_n :=(\rho^{\otimes n} - 2^{\lambda n}\sigma^{\otimes n})$ allows us to write it in the Schur basis as
\begin{equation}
R_n =  \bigoplus_{\lambda} R_{\lambda} \otimes \id_{\lambda},
\end{equation}
where, as in the proof of Lemma \ref{ulhmansymmetric}, $\lambda$ labels the irreps of $S_n$, $\id_{\lambda}$ is the identity on the irrep labelled by $\lambda$, and $R_{\lambda}$ is a Hermitian operator acting on the multiplicity space of the the irrep labelled by $\lambda$ \cite{FH91}. It is then clear that
\begin{equation}
Q_n =  \bigoplus_{\lambda} Q_{\lambda} \otimes \id_{\lambda},
\end{equation}
where the $Q_{\lambda}$ are projectors onto $(R_{\lambda})_+$. Likewise,  
\begin{equation}
\sigma^{\otimes n} =  \bigoplus_{\lambda} \sigma_{\lambda} \otimes \id_{\lambda}, \hspace{0.2 cm} \rho^{\otimes n} =  \bigoplus_{\lambda} \rho_{\lambda} \otimes \id_{\lambda}
\end{equation} 
for positive semidefinite operators $\sigma_\lambda, \rho_{\lambda}$. 

As $\text{supp}(R_n) \subseteq \text{supp}(\sigma^{\otimes n})$, we have that for each $\lambda$, $\text{supp}(R_{\lambda}) \subseteq \text{supp}(\sigma_{\lambda})$. We consider an eigen-decomposition of $R_{\lambda} := \sum_k e_{k, \lambda} E_{k, \lambda}$ with eigenprojectors $E_{k, \lambda}$ divided into three disjoint subsets, with members of the first one being subprojections of $\text{supp}(R_{\lambda})$, members of the second one being subprojections of the orthogonal complement of $\text{supp}(R_{\lambda})$ in $\text{supp}(\sigma_{\lambda})$, and members of the third one being subprojections of $\text{supp}(\sigma_{\lambda})^{\bot}$. Defining the quantum operation ${\cal E}_{\lambda}(X) := \sum_k E_{k, \lambda}XE_{k, \lambda}$, this particular choice of eigen-projectors $E_{k, \lambda}$ ensures that $\text{supp}({\cal E}_{\lambda}(\sigma_{\lambda})) \subseteq \text{supp}(\sigma_{\lambda})$, a property which will be used next.  

We identify the original eigen-projectors $\{ E_{k} \}$ of $Q$ with $\{ \bigoplus_{\lambda} E_{\lambda, k_{\lambda}} \otimes \id_{\lambda} \}$, for all possible combinations of the labels $k,{\lambda}$. Then ${\cal E}(X) = \bigoplus_{\lambda} {\cal E}_{\lambda} \otimes \id_{\lambda}(X)$ and we can write
\begin{eqnarray} \label{bbound222}
\sum_i p_n(i)q_n(i)^{-s} &=& \tr({\cal E}(\rho^{\otimes n})({\cal E}(\sigma^{\otimes n}))^{-s}) \nonumber \\
 &=& \tr(\rho^{\otimes n}({\cal E}(\sigma^{\otimes n}))^{-s})
 \nonumber \\ &=& \sum_{\lambda} \tr(\rho_{\lambda}{\cal E}_{\lambda}(\sigma_{\lambda})^{-s}) \dim(\id_{\lambda}).
\end{eqnarray}
From Lemma 9 of Ref. \cite{Hay02} we find for each $\lambda$, $\sigma_{\lambda} \leq \dim{{\cal H}_{\lambda}}{\cal E}_{\lambda}(\sigma_{\lambda})$, where ${\cal H}_{\lambda}$ is the Hilbert space in which $\sigma_{\lambda}$ acts on. As $\text{supp}({\cal E}_{\lambda}(\sigma_{\lambda})) = \text{supp}(\sigma_{\lambda})$, we can apply the operator monotonicity of $-u^{-1}$ for $0 < t \leq 1$ to get
\begin{equation}
({\cal E}_{\lambda}(\sigma_{\lambda}))^{-s} \leq  (\dim{{\cal H}_{\lambda}})^{s}(\sigma_{\lambda})^{-s}.
\end{equation}
Applying the equation above to Eq. (\ref{bbound222}) and using the bound $\dim({\cal H}_{\lambda}) \leq (n + 1)^{\dim({\cal H})}$ on the dimension of the multiplicity spaces ${\cal H}_{\lambda}$ \cite{FH91}, 
\begin{eqnarray} 
\sum_i p_n(i)q_n(i)^{-s} &\leq&  (n + 1)^{s\dim({\cal H})} \sum_{\lambda} \tr(\rho_{\lambda}(\sigma_{\lambda})^{-s}) \dim(\id_{\lambda}) \nonumber \\ &=& (n + 1)^{s\dim({\cal H})} \tr(\rho^{\otimes n}(\sigma^{\otimes n})^{-s}),
\end{eqnarray} 
and we are done.
\end{proof}

\vspace{0.5 cm}

We are now in position to prove the direct part of Proposition \ref{maincompact}.

\vspace{0.8 cm}

\begin{proof} (Direct part Proposition \ref{maincompact})

We show that
\begin{equation}
        \liminf_{n \rightarrow \infty} \min_{\omega_{n} \in {\cal M}_n} \tr( \rho^{\otimes n} - 
        2^{yn}\omega_n)_+ = 1 - \mu,
\end{equation}
with $\mu > 0$, implies $y \geq E_{\cal M}^{\infty}(\rho)$. First, if $\mu = 1$, we find from Corollary \ref{maincompactstrongconverse} that $y > E_{\cal M}^{\infty}(\rho)$. So in the rest of the proof we show that if $0 < \mu < 1$, then $y \geq E_{{\cal M}}^{\infty}(\rho)$. 

Let $\{ \sigma_n \in {\cal M}_n \}_{n \in \mathbb{N}}$ be a sequence of 
optimal solutions in the minimization of Eq. (\ref{otherpart}). Note that from Lemma \ref{monotonicity3} and property \ref{cond5} of 
the sets $\{ {\cal M}_n \}_{n \in \mathbb{N}}$, we can take the states 
$\sigma_n$ to be permutation-symmetric.
 
For each $n \in \mathbb{N}$ we have $\rho^{\otimes n} \leq 2^{yn}\sigma_n 
+ ( \rho^{\otimes n} - 2^{yn}\sigma_n)_+$. Applying Lemma \ref{DR} once 
more we see that there is an increasing sequence ${\cal F}$ of the integers going to infinity and states $\rho_n$, with $n \in {\cal F}$, such that $F(\rho_n, \rho^{\otimes n}) \geq 
\mu/2 := \lambda$ and  
\begin{equation}\label{fidelitynext}
        \rho_n \leq  \frac{2^{yn}}{\lambda}\sigma_n,
\end{equation}

From Lemma \ref{monotonicity3} and the permutation-invariance of 
$\sigma_n$ and $\rho^{\otimes n}$, we can also take $\rho_n$ to be 
permutation-symmetric. Let $\ket{\theta} \in {\cal H}\otimes {\cal H}_E$ 
be a purification of $\rho$, where ${\cal H}_E \cong {\cal H}$ is the purifying Hilbert space. Then, by Lemma \ref{ulhmansymmetric} there 
is a permutation-symmetric purification $\ket{\Psi_n} \in {\cal H}^{\otimes n} \otimes {\cal H}_E^{\otimes n}$ of $\rho_n$ 
such that $|\braket{\theta^{\otimes n}}{\Psi_n}| \geq \lambda$. By 
Lemma \ref{variantexpdefineti} and Eq. (\ref{fidelitynext}), in turn, 
we find that there is a $\ket{\Psi_{n, m}}$ approximating 
$\ket{\Psi_{n, m, r}} \in \ket{\theta}^{[\otimes, n - m, r]}$ such that
\begin{equation}
        || \ket{\Psi_{n, m}}\bra{\Psi_{n, m}} - \ket{\Psi_{n, m, r}}\bra{\Psi_{n, m, r}} ||_1 
        \leq 2\sqrt{2} \lambda^{-1}  e^{-\frac{mr}{2n}}
\end{equation}
and
\begin{equation}
        \tr_{E}(\ket{\Psi_{n, m}}\bra{\Psi_{n, m}}) \leq \lambda^{-2}
        \tr_{1, ..., m}(\rho_n) 
        \leq  \lambda^{- 3 }2^{yn} \tr_{1, ..., m}(\sigma_n),
\end{equation}
where the partial trace is taken over the purifying Hilbert space ${\cal H}_E^{\otimes n - m}$. 

From the operator monotonicity of the $\log$ and property \ref{cond3} of the sets, 
\begin{equation}
        \frac{1}{n} E_{{\cal M}_{n-m}}(\tr_{E}(\ket{\Psi_{n, m}}\bra{\Psi_{n, m}})) 
        \leq y -  3 \frac{\log(\lambda)}{n}
\end{equation}
From Lemma \ref{asympcont}
\begin{equation} \label{asymptbeforemayelast}
        \frac{1}{n} E_{{\cal M}_{n-m}}(\tr_{E}(\ket{\Psi_{n, m, r}}\bra{\Psi_{n, m, r}})) 
        \leq y -  3 \frac{\log(\lambda)}{n} + 
        f(2\sqrt{2} \lambda^{-1} e^{-\frac{mr}{2n}})
\end{equation}
for every $r \leq n - m$, where $f : \mathbb{R} \rightarrow \mathbb{R}$ is such that $\lim_{x \rightarrow 0} f(x) = 0$.

Then, setting $m = r = n^{2/3}$, taking the limit $n \rightarrow \infty$ in Eq. (\ref{asymptbeforemayelast}), and using Lemma \ref{mainpart2}, we find that for every $\rho$ with $\lambda_{\max}(\rho) < 1$,
\begin{equation}
E_{{\cal M}}^{\infty}(\rho) \leq \liminf_{n \rightarrow \infty} \frac{1}{n} E_{{\cal M}_{n-m}}(\tr_{E}(\ket{\Psi_{n, m, r}}\bra{\Psi_{n, m, r}})) \leq y.
\end{equation}

Finally, we show that the result for non-pure states implies its validity to pure states too, completing the proof. Let $\ket{\psi}$ be 
a pure state and $y < E_{\cal M}^{\infty}(\ket{\psi}\bra{\psi})$. Asymptotic continuity of $E_{\cal M}^{\infty}$ (see Lemma \ref{asympcont}) yields the existence of a $\chi > 0$ such that $y < E_{\cal M}^{\infty}(\zeta)$ for $\zeta := (\ket{\psi}\bra{\psi} + \chi \sigma)/(1 + \chi)$, where $\sigma$ is the full rank state from property \ref{cond2} of the sets ${\cal M}_n$. Then, assuming the result for mixed states, we have 
\begin{equation} \label{sdsgfsdfgfosn}
\lim_{n \rightarrow \infty} \min_{\omega_n \in {\cal M}_n} \tr(\zeta^{\otimes n} - 2^{yn}\omega_n)_+ = 1.
\end{equation}

By the asymptotic equipartition theorem \cite{CT91} we can find a sequence of states $\zeta_n = \sum_i p_{i, n} \zeta_{i, n}$ where $\{ p_{i, n} \}$ is a probability distribution and each $\zeta_{i, n}$ is - up to permutations of the copies - of the form $(\ket{\psi}\bra{\psi})^{\otimes n - m_{i, n}} \otimes \sigma^{\otimes m_{i, n}}$, with 
\begin{equation}
\lim_{n \rightarrow \infty} \max_{i} \frac{m_{i, n}}{n} = \lim_{n \rightarrow \infty} \min_{i} \frac{m_{i, n}}{n} = \chi/(1 + \chi)
\end{equation}
and $\lim_{n \rightarrow \infty} ||\zeta^{\otimes n} - \zeta_n ||_1 = 0$. In particular the inequality $\tr(\zeta^{\otimes n} - 2^{yn}\omega_n)_+ \leq \tr(\zeta_n - 2^{yn}\omega_n)_+ + || \zeta^{\otimes n} - \zeta_n ||_1$ yields 
\begin{equation}
\lim_{n \rightarrow \infty} \min_{\omega \in {\cal M}_n} \tr(\zeta_n - 2^{yn})_+ = 1.
\end{equation}

Note also that $(X, Y) \mapsto \tr(X - Y)_+$ is convex and hence $\rho \mapsto \min_{\omega_n \in {\cal M}_n} \tr(\rho - 2^{yn}\omega_n)_+$ is convex too. Therefore
\begin{equation}
\min_{\omega_n \in {\cal M}_n} \tr(\zeta_n - 2^{y n} \omega_n)_+ \leq \sum_i p_{i, n} \min_{\omega_n \in {\cal M}_n} \tr(\zeta_{i, n} - 2^{yn}\omega_n)_+ \leq \max_i \min_{\omega_n \in {\cal M}_n} \tr(\zeta_{i, n} - 2^{y n}\omega_n)_+.
\end{equation}
Let $i^*$ be a maximizer of the last formula above. Then, $\zeta_{i^*, n}$ can be written as $P_{f_{i^*}}(\ket{\psi}\bra{\psi}^{\otimes n - m_n} \otimes \sigma^{\otimes m})P_{f_{i^*}}^*$, for some $ m = m(n) \in \mathbb{N}$ and $f_{i^*} \in S_n$. Hence
\begin{eqnarray}
\max_i \min_{\omega_n \in {\cal M}_n} \tr(\zeta_{i, n} - 2^{y n}\omega_n)_+ &\leq& \min_{\omega_n \in {\cal M}_{n - m}}\tr(P_{f_{^{i*}}}(\ket{\psi}\bra{\psi}^{\otimes n - m}\otimes \sigma^{\otimes m})P_{f_{^{i*}}}^* - P_{f_{^{i*}}}(\omega_n \otimes \sigma^{\otimes m})P_{f_{^{i*}}}^*) \nonumber \\ &=& \min_{\omega_n \in {\cal M}_{n - m}} \tr(\ket{\psi}\bra{\psi}^{\otimes n - m} - 2^{y n}\omega_n)_+.  
\end{eqnarray}
By the above, 
\begin{equation}
1 \leq \liminf_{n \rightarrow \infty} \min_{\omega \in {\cal M}_{n - m}} \tr(\ket{\psi}\bra{\psi}^{\otimes n - m} - 2^{y n}\omega_n)_+ \leq \liminf_{n \rightarrow \infty} \min_{\omega \in {\cal M}_{n}} \tr(\ket{\psi}\bra{\psi}^{\otimes n} - 2^{y n}\omega_n)_+,
\end{equation}
where in the last inequality we used that $\lim_{n \rightarrow \infty} n - m = +\infty$, due to the assumption $\lim_{n \rightarrow \infty} \frac{1}{n} \max_i m_{i, n} = \frac{\chi}{1 + \chi}$.
\end{proof}

\vspace{0.8 cm}

The next lemma shows a property of the measures $E_{{\cal M}_k}$ analogous to the non-lockability of the 
relative entropy of entanglement \cite{HHHO05}, in this case manifested in the almost power states. 

\begin{lemma} \label{mainpart2}
Let $\ket{\theta} \in {\cal H}\otimes {\cal H}_E$ and $\rho = \tr_{E}(\ket{\theta}\bra{\theta})$ with $\lambda_{\max}(\rho) < 1$. Let $\{ \ket{\Psi_{n, m, r}} \in \ket{\theta}^{[\otimes, n - m, r]} \}_{n, m, r}$ be a sequence of almost power states along $\ket{\theta}$, with $r = o(n)$ and $m = o(n)$. Then
\begin{equation} \label{nonlockability}
E_{{\cal M}}^{\infty}(\rho) \leq \liminf_{n \rightarrow \infty} \frac{1}{n} E_{{\cal M}_{n-m}}(\tr_{E}(\ket{\Psi_{n, m, r}}\bra{\Psi_{n, m, r}})).
\end{equation}
\end{lemma}

\begin{proof}
Write $\ket{\Psi_{n, m, r}} 
= \sum_{k=0}^{r} \beta_k \text{Sym}(\ket{\eta_k} \otimes \ket{\theta}^{\otimes n - m - k})$, 
where $\ket{\eta_k}$ are permutation-symmetric states living in $({\cal H} \bot \ket{\theta})^{\otimes k}$ and 
$\sum_k |\beta_k|^2 = 1$. Define
\begin{equation} \label{approx1}
\ket{{\Phi}_{n, m, r}} := \sum_{k : |\beta_k| \geq 1/n} \beta_k \text{Sym}(\ket{\eta_k} \otimes \ket{\theta}^{\otimes n - m - k})
\end{equation}
and $\ket{\tilde{\Phi}_{n, m, r}} := \ket{\Phi_{n, m, r}} / || \ket{\Phi_{n, m, r}} ||$. Note that $\lim_{n \rightarrow \infty} ||\ket{\tilde{\Phi}_{n, m, r}} - \ket{\Psi_{n, m, r}}|| = 0$. Thus, from the asymptotic continuity of the measures $E_{{\cal M}_k}$ (Lemma \ref{asympcont}) it follows
\begin{equation} \label{agcdnsfrhpsiohfkshf}
\liminf_{n \rightarrow \infty} \frac{1}{n} E_{{\cal M}_{n-m}}(\tr_{E}(\ket{\Psi_{n, m, r}}\bra{\Psi_{n, m, r}})) = \liminf_{n \rightarrow \infty} \frac{1}{n} E_{{\cal M}_{n-m}}(\tr_{E}(\ket{\tilde{\Phi}_{n, m, r}}\bra{\tilde{\Phi}_{n, m, r}})),
\end{equation}
and thus it suffices to show that the R.H.S. of the equation above is larger or equal to $E_{{\cal M}}^{\infty}(\rho)$.

From Lemma \ref{trickwithsoivadfsadfsadf} we find 
\begin{eqnarray} \label{xuvjwovfd}
        (\ket{\theta}\bra{\theta})^{\otimes n - m - r}
        &\leq&  2^{nh\left(\frac{r}{n - m}\right)}n^2 \tr_{1,...,r}(\ket{{\Phi}_{n, m, r}}\bra{{\Phi}_{n, m, r}}) \nonumber \\
        &\leq& 2^{nh\left(\frac{r}{n - m}\right)}n^2 \tr_{1,...,r}(\ket{\tilde{\Phi}_{n, m, r}}\bra{\tilde{\Phi}_{n, m, r}}),
\end{eqnarray}
where the last inequality follows from $|| \ket{\Phi_{n, m, r}} || \leq 1$.

For simplicity of notation we define $\pi_n := \tr_{1,...,r}\tr_E(\ket{\tilde{\Phi}_{n, m, r}}\bra{\tilde{\Phi}_{n, m, r}})$. Tracing out the environment Hilbert space in Eq. (\ref{xuvjwovfd}),
\begin{eqnarray}
        \rho^{\otimes n - m - r} \leq 2^{nh\left(\frac{r}{n - m}\right)}n^2 \pi_n.
\end{eqnarray}

Let $\tilde \omega_n \in {\cal M}_{n - m - r}$ be such that
\begin{equation} 
        E_{{\cal M}_{n - m - r}}(\pi_n) = S(\pi_n || \tilde \omega_n).
\end{equation}
and set
\begin{equation} 
\omega_n := \frac{1}{1 + \tau}\tilde \omega_n + \frac{\tau}{1 + \tau}\sigma^{\otimes n - m - r},
\end{equation}
where and $\tau > 0$. We introduce $\omega_n$ in order to have a non-negligible lower bound on the minimum eigenvalue of a close-to-optimal state for $\pi_n$, which will show useful later on. 

From the previous equation and the operator monotonicity of the $\log$ function,
\begin{equation} 
        E_{{\cal M}_{n - m - r}}(\pi_{n}) = S(\pi_n || \tilde \omega_n) \geq S(\pi_n || \omega_n) - \log(1 + \tau).
\end{equation}

Let $\lambda_{n, \nu} = E_{{\cal M}_{n - m - r}}(\pi_n) + n\nu + \log(1 + \tau) \geq S(\pi_n || \omega_n) + n \nu$, for $\nu > 0$. For every integer $l$
\begin{eqnarray} \label{beforelast}
        \rho^{\otimes  (n - m - r)l} &\leq& n^{2l}
        2^{nh\left(\frac{r}{n - m}\right)l}\pi_n^{\otimes l} \nonumber \\ 
        &\leq& n^{2l}2^{nh\left(\frac{r}{n - m}\right)l}2^{\lambda_{n, \nu} l} \omega_n^{\otimes l} 
        + n^{2l}2^{nh\left(\frac{r}{n - m}\right)l}(\pi_n^{\otimes l} - 2^{\lambda_{n, \nu} l} 
        \omega_n^{\otimes l})_+.
\end{eqnarray}

From Lemma \ref{mainpart3} we find that for every $\nu > 0$ there is a constant $\gamma > 0$ with the property that for every $n \in \mathbb{N}$, there is an integer $l_{n}$ such that 
\begin{equation} \label{finalfinal}
        \tr(\pi_n^{\otimes l} - 2^{\lambda_{n, \nu} l} 
        \omega_n^{\otimes l})_+ \leq 2^{- \gamma n l}.
\end{equation}
for every $l \geq l_{n}$. 

Then applying Lemma \ref{DR} to Eq. (\ref{beforelast}), we find that for every $n$ sufficiently large, there is a sequence of states 
$\rho_{l, n}$ such that $\lim_{l \rightarrow \infty} || \rho_{l, n} -  \rho^{\otimes (n - m - r)l} ||_1 = 0$ and
\begin{equation}
        \rho_{l, n} \leq g(l)(n^2 2^{nh\left(\frac{r}{n-m}\right)})^l 2^{\lambda_{n, \nu} l} 
        \omega_n^{\otimes l}, 
\end{equation}
for a function $g(l)$ such that $\lim_{l \rightarrow \infty}g(l) = 1$. Then we have
\begin{eqnarray}
(n - m - r)E_{\cal M}^{\infty}(\rho) &=& E_{\cal M}^{\infty}(\rho^{\otimes n - m - r}) = \lim_{l \rightarrow \infty} \frac{1}{l} E_{{\cal M}_{(n - m - r)l}} (\rho^{\otimes (n - m - r)l}) \nonumber \\ &=& \lim_{l \rightarrow \infty} \frac{1}{l} E_{{\cal M}_{(n - m - r)l}}(\rho_{l, n}) \leq \lim_{l \rightarrow \infty} \frac{1}{l} S_{\max}(\rho_{l, n} || \omega_n^{\otimes l}) \nonumber \\ &\leq& \lim_{l \rightarrow \infty} \frac{1}{l} \log g(l) + 2 \log(n) + n h \left( \frac{r}{n - m} \right ) + \lambda_{n, \nu} \nonumber \\ &=& 2 \log(n) + n h \left( \frac{r}{n - m} \right ) + E_{{\cal M}_{n - m - r}}(\pi_n) + \nu n + \log(1 + \tau)
\end{eqnarray}
and, since, $E_{{\cal M}_{n - m - r}}(\pi_n) \leq E_{{\cal M}_{n - m}}(\tr_E(\ket{\tilde \Phi_{n, m, r}}\bra{\tilde \Phi_{n, m, r}}))$,
\begin{eqnarray}
E_{\cal M}^{\infty}(\rho) &=& \liminf_{n \rightarrow \infty} \frac{1}{n - m - r} \left( 2 \log(n) + n h \left( \frac{r}{n - m} \right ) + E_{{\cal M}_{n - m - r}}(\pi_n) + \nu n + \log(1 + \tau) \right) \nonumber \\ &\leq& \liminf_{n \rightarrow \infty} E_{{\cal M}_{n - m}}(\tr_E(\ket{\tilde \Phi_{n, m, r}}\bra{\tilde \Phi_{n, m, r}})) + 2 \nu. \nonumber
\end{eqnarray}
Taking $\nu$ to zero and using Eq. (\ref{agcdnsfrhpsiohfkshf}) we find Eq. (\ref{nonlockability}). 
\end{proof}

\vspace{0.8 cm}

As in the proof above, let $\ket{\theta} \in {\cal H} \otimes {\cal H}_E$ and $\rho := \tr_{E}(\ket{\theta}\bra{\theta})$ be such that $\lambda_{\max}(\rho) < 1$. The next three lemmata concern the following states:
\begin{equation} \label{afdahfdkajhfdkjahdfkjahdfASD}
\ket{\Phi_{n, m, r}} := \sum_{k : |\beta_k| \geq 1/n} \beta_k \text{Sym}(\ket{\eta_k} \otimes \ket{\theta}^{\otimes n - m - k}),
\end{equation}
for complex-valued coefficients $\beta_k$ and states $\ket{\eta_k}$ living in $({\cal H} \bot \ket{\theta})^{\otimes k}$, and 
\begin{equation} \label{pipipipin}
\pi_n := \tr_{1,...,r}\tr_E(\ket{\Phi_{n, m, r}}\bra{\Phi_{n, m, r}}) / \braket{\Phi_{n, m, r}}{\Phi_{n, m, r}}.
\end{equation}

\begin{lemma} \label{trickwithsoivadfsadfsadf} 
Let $k_{\max} \leq (n-m)/2$ be the maximum $k$ appearing in Eq. (\ref{afdahfdkajhfdkjahdfkjahdfASD}). Then, for $r \geq k_{\max}$,
\begin{eqnarray} 
        (\ket{\theta}\bra{\theta})^{\otimes n - m - r}
        &\leq&  2^{nh\left(\frac{r}{n - m}\right)}n^2 \tr_{1,...,r}(\ket{{\Phi}_{n, m, r}}\bra{{\Phi}_{n, m, r}}),
\end{eqnarray}
\end{lemma}

\begin{proof}
Let $\ket{\phi} := \ket{\eta_{k_{\max}}} \otimes \ket{\theta}^{\otimes n - m - k_{\max}}$. Then 
\begin{equation} \label{,mxzbcv,mxzbvc,mzcvvz}
\ket{{\Phi}_{n, m, r}} = c\ket{\phi} + c' e^{i\vartheta}\ket{\phi^{\bot}},
\end{equation}
where
\begin{equation} \label{cccccccccccccccccc}
c := \binom{n - m}{k_{\max}}^{-1/2}\beta_{k_{\max}},
\end{equation}
$\vartheta \in \mathbb{R}$, $c' \geq 0$, and $\ket{\phi^{\bot}}$ is a state orthogonal to $\ket{\phi}$. From Eq. (\ref{afdahfdkajhfdkjahdfkjahdfASD}), we can write $\ket{\phi^{\bot}}$ as a superposition of states of the form $\ket{f_1} \otimes ... \otimes \ket{f_{n - m}}$, where at least in one of the first $k_{\max}$ registers, $\ket{f_i} = \ket{\theta}$. Therefore, as $\ket{\eta_{k_{\max}}}$ lives in $({\cal H} \bot \ket{\theta})^{\otimes k_{\max}}$, we get $\tr_{1, ..., k_{\max}}(\ket{\phi}\bra{\phi^{\bot}}) = 0$ and thus
\begin{eqnarray}
\tr_{1, ... k_{\max}}(\ket{{\Phi}_{n, m, r}}\bra{{\Phi}_{n, m, r}}) &=& |c|^2 \tr_{1, ... k_{\max}}(\ket{\phi}\bra{\phi}) + (c')^2\tr_{1, ... k_{\max}}(\ket{\phi^{\bot}}\bra{\phi^{\bot}}) \nonumber \\ &\geq& |c|^2 \tr_{1, ... k_{\max}}(\ket{\phi}\bra{\phi}) \nonumber \\ &=& |c|^2 (\ket{\theta}\bra{\theta})^{\otimes n - m - k_{\max}}.
\end{eqnarray}
From Eq. (\ref{cccccccccccccccccc}),
\begin{equation}
(\ket{\theta}\bra{\theta})^{\otimes n - m - k_{\max}} \leq \binom{n - m}{k_{\max}}|\beta_{k_{\max}}|^{-2} \tr_{1, ... k_{\max}}(\ket{{\Phi}_{n, m, r}}\bra{{\Phi}_{n, m, r}}).
\end{equation}

Note that $|\beta_{k_{\max}}|^{-2} \leq n^{2}$ and the entropic bound $\binom{n}{k} \leq 2^{nh(k/n)}$ (see e.g. Lemma 17.5.1 of \cite{CT91}). Moreover, from the monotonicity of the binary entropy in the interval $[0, 1/2]$, $h(k_{\max}/(n - m)) \leq h(r/(n - m))$. Therefore, 
\begin{equation}
(\ket{\theta}\bra{\theta})^{\otimes n - m - k_{\max}} \leq 2^{nh\left(\frac{r}{n - m}\right)}n^2 \tr_{1, ... k_{\max}}(\ket{{\Phi}_{n, m, r}}\bra{{\Phi}_{n, m, r}}).
\end{equation}
The lemma follows by tracing out the first $r - k_{\max}$ registers in the equation above.
\end{proof}

\vspace{0.8 cm}

As in the proof of the direct part of Proposition \ref{maincompact}, let $\tilde{\omega}_n$ be such that $E_{{\cal M}_{n - m - r}}(\pi_n) = S(\pi_n || \tilde \omega_n)$ and define 
\begin{equation} \label{tildeomega}
\omega_n := \frac{1}{1 + \tau}\tilde \omega_n + \frac{\tau}{1 + \tau}\sigma^{\otimes n - m - r},
\end{equation}
with $\tau > 0$.

\begin{lemma} \label{mainpart3}
Let $\omega_n$ be given by Eq. \ref{tildeomega}, $\pi_n$ by Eq. (\ref{pipipipin}), and $\lambda$ be such that 
\begin{equation}
\lambda = \lambda_{n, \nu} \geq S(\pi_n || \omega_n) + \nu n,
\end{equation}
for $\nu > 0$. Then, there is a $\gamma > 0$ and a sequence $\{ l_n \}_{n \in \mathbb{N}}$ such that for sufficiently large $n$ and $l \geq l_{n}$, 
\begin{equation} \label{finalfinala}
        \tr(\pi_n^{\otimes l} - 2^{\lambda_{n, \nu} l} 
        \omega_n^{\otimes l})_+ \leq 2^{- \gamma n l},
\end{equation}
\end{lemma}
\begin{proof}
From Lemma \ref{ogawanagaokaadaptation},
\begin{eqnarray} \label{boundstwo}
\tr(\pi_n^{\otimes l} - 2^{\lambda l} \omega_n^{\otimes l})_+ \leq 2^{-lp(s)} + 2^{-lq(s)}, 
\end{eqnarray}
with $p_n(s) := (s\mu - \log \tr(\pi_n^{1 + s}))$ and $q_n(s) := (s(\lambda - \mu) - s D^{n - m - r} \frac{\log(1 + l)}{l} - \log \tr(\pi_n \omega_n^{-s}))$. We set $\mu = (\nu/2 - S(\rho))n$ and show that each of the two bounds in the equation above is smaller than $2^{-\gamma n l}$, for a given constant $\gamma$ and sufficiently large $n$ and $l \geq l_{n}$. 

From Eq. (\ref{approx1}) we can write $\pi_n = \tr_{1, ..., r}\tr_E(\ket{\Psi_{\pi_n}}\bra{\Psi_{\pi_n}})$ (identifying $\ket{\Psi_{\pi_n}}$ and $\ket{\Phi_{n, m, r}} / || \ket{\Phi_{n, m, r}} ||$), with
\begin{equation}
\ket{\Psi_{\pi_n}} := \sum_{k=0}^r \alpha_k \text{Sym}(\ket{\chi_k} \otimes \ket{\theta}^{\otimes n - m - k}),
\end{equation} 
where $\sum_{k=0}^r |\alpha_k|^2 = 1$ and 
\begin{equation} \label{kasfksajhfdshdflsajfdalkfd}
\ket{\chi_k} \in ({\cal H} \bot \ket{\theta})^{\otimes k}. 
\end{equation}
Each $\text{Sym}(\ket{\chi_k} \otimes \ket{\theta}^{\otimes n - m - k})$ is a superposition of $\binom{n - m}{k}$ terms which, up to permutation of the copies and normalization, have the form $\ket{\chi_k} \otimes \ket{\theta}^{\otimes n - m - k}$; let us denote these by $\ket{\psi_{k, j}}$. from Eq. (\ref{kasfksajhfdshdflsajfdalkfd}), we get $|\braket{\psi_{k, j}}{\psi_{k', j'}}| = \delta_{kk'}\delta_{jj'}$. Therefore we can write 
\begin{equation}
\ket{\Psi_{\pi_n}} = \sum_{k=0}^r \sum_{j=1}^{\binom{n-m}{k}} \varsigma_{k, j} \ket{\psi_{k, j}}, 
\end{equation} 
with $\sum_{k, j}|\varsigma_{k, j}|^2 = 1$. By Lemma \ref{supversusconvex},
\begin{equation}
\ket{\Psi_{\pi_n}} \bra{\Psi_{\pi_n}} \leq  (r+1)\binom{n-m}{r}\sum_{k, j} |\varsigma_{k, j}|^2 \ket{\psi_{k, j}}\bra{\psi_{k, j}},
\end{equation}
where we used that since $k, m, r = o(n)$, $\binom{n-m}{k} \leq \binom{n-m}{r}$ for every $k \leq r$. Tracing out $E$ and the first $r$ copies in
both sides of the equation above, we find
\begin{equation} \label{supconv}
\pi_n \leq (r + 1) \binom{n - m}{r} \sum_{j} p_j \rho_j \leq (r + 1)2^{(n-m) h\left( \frac{r}{n-m} \right)} \sum_j p_j \rho_j,  
\end{equation}
where $\{ p_j \}$ is a probability distribution and each $\rho_{j}$ is of the form $\rho^{\otimes n - m - r} \otimes \sigma_r$, up to permutations of the copies, with an arbitrary state $\sigma_r$ acting on ${\cal H}^{\otimes r}$. 

Then, by the Schur-convexity of the function $h(x) = x^{1 + s}$ ($s \geq 0$),
\begin{eqnarray}
\tr(\pi_n^{1 + s}) &\leq& (r+1)^{1 + s}2^{(n-m) h\left( \frac{r}{n-m} \right)(1 + s)} \tr((\sum_j p_j \rho_j)^{1 + s}) \nonumber \\ &\leq& (r+1)^{1 + s}2^{(n-m) h\left( \frac{r}{n-m} \right)(1 + s)} \sum_j p_j \tr(\rho_j^{1 + s}),
\end{eqnarray}
from which follows that, with $h_{n, m, r, s} := - (1 + s)(\log(r + 1) + (n-m) h\left( \frac{r}{n-m} \right))$,
\begin{eqnarray}
- \log \tr(\pi_n^{1 + s}) &\geq& h_{n, m, r, s} - \max_j \log \tr(\rho_j^{1 + s}) \nonumber \\ &=& h_{n, m, r, s} - \max_j \log \tr((\sigma_j)^{1 + s}) - (n - m - r)\log \tr (\rho^{1 + s}) \nonumber \\  &\geq& h_{n, m, r, s} + (m + r)\log \tr (\rho^{1 + s}) - n\log \tr (\rho^{1 + s}),
\end{eqnarray}
where the last inequality follows from $\tr((\sigma_j)^{1 + s}) \leq 1$. Note that the first two terms in the equation above are $o(n)$. Therefore
\begin{equation}
- \log \tr(\pi_n^{1 + s}) \geq - n \log \tr(\rho^{1 + s}) - o(n).
\end{equation}

Letting $g(s) := - \log \tr (\rho^{1 + s})$, we see that $g(0) = 0$ and 
$g'(0) = S(\rho)$. Then, 
\begin{eqnarray}
p_n(s) &=& s(\nu/2 - S(\rho))n - \log \tr(\pi_n^{1+s}) \nonumber \\ &\geq& s(\nu/2 - S(\rho))n - n \log \tr(\rho^{1 + s}) - o(n) \nonumber \\
&\geq& n s\nu/2 - n \max_{0\leq t \leq s}| g''(t)| s^2 - o(n)  .     
\end{eqnarray}
Thus there is a $s$ small enough, independent of $n$, such that for sufficiently large $n$, $p_{n}(s) \geq n s \nu/4$.

Considering the second bound in Eq. (\ref{boundstwo}), let $f_n(s) := - \frac{1}{n}\log \tr \left( \pi_n \omega_n^{-s} \right)$. As $\omega_n$ is full rank, we find from Taylor's Theorem, 
\begin{equation}
- \frac{1}{n}\log \tr \left( \pi_n \omega_n^{-s} \right) = f_n(0) + f_n'(0)s + f_n''(t_{s, n})s^2/2,
\end{equation}
for some real number $t_{s, n} \leq s$. A simple calculation shows that $f_n(0) = 0$, 
\begin{equation}
        f_n'(0) =  \frac{1}{n} \tr(\pi_n \log  \omega_n),
\end{equation}
and
\begin{equation} \label{f''}
        f_n''(s) = - \frac{1}{n} \left( \frac{\tr(\pi_n \omega_n^{-s} (\log \omega_n)^2)}
        {\tr(\pi_n \omega_n^{-s})} - \left( \frac{\tr(\pi_n  \omega_n^{-s} \log \omega_n)}
        {\tr(\pi_n \omega_n^{-s})}\right)^2 \right).
\end{equation}
We next show that there is a $s$ sufficiently small, but independent of $n$, such that
\begin{equation} \label{sup0}
\max_{0 \leq t \leq s}| f_n''(t) | \leq 1 
\end{equation}
for $n$ sufficiently large. Hence
\begin{eqnarray}
q_n(s) &\geq& s(n\nu/2 + S(\pi_n || \omega_n) + nS(\rho) + \tr(\pi_n \log \omega_n))  - s D^{n - m - r} \frac{\log(1 + l)}{l} - n \max_{0 \leq t \leq s}| f_n''(t) | s^2   \nonumber \\
&\geq& \frac{s \nu n}{2} + s(nS(\rho) - S(\pi_n)) - s D^{n - m - r} \frac{\log(1 + l)}{l} - ns^2.
\end{eqnarray}
Using Lemma \ref{nonlockentropyalmostppowerstates}, choosing $s$ sufficiently small and $l_n$ such that $D^{n - m - r} \frac{\log(1 + l_n)}{l_n} = o(n)$, we find $q_n(s) \geq n s \nu/4$, for sufficiently large $n$ and $l \geq l_n$. 

In order to prove Eq. (\ref{sup0}), we consider the basis where 
$\pi_n$ is diagonal  
\begin{equation}
        \pi_n = \text{Diag}(\lambda_{1, n}, \lambda_{2, n}, ...).
\end{equation}
and write $\omega_n$ in this basis 
\begin{equation}
        \omega_n = U \text{Diag}
        (\mu_{1, n}, \mu_{2, n}, ...) U^{\cal y},
\end{equation}
where $U$ is a unitary. Note that Eq. (\ref{tildeomega}) gives 
\begin{equation} \label{nonameuntilnow}
\omega_n = \frac{1}{1 + \tau}\tilde \omega_n + \frac{\tau}{1 + \tau}\sigma^{\otimes n - m - r} \geq \frac{\tau}{1 + \tau}\sigma^{\otimes n - m - r} \geq \frac{\tau}{1 + \tau} \lambda_{\min}(\sigma)^{n - m - r}.
\end{equation}
where $\lambda_{\min}(\sigma) > 0$ is the minimum eigenvalue of $\sigma$. 

From Eq. (\ref{f''}) it follows that we can write
\begin{equation}
|f_n''(s)| = \frac{1}{n} \left( \sum_{j} t_{j, n} (\log \mu_{j, n})^2 - \left( \sum_{j} t_{j, n} \log \mu_{j, n}   \right)^2 \right),
\end{equation}
where $\{ t_{j, n} \}$ is the probability distribution given by
\begin{equation}
t_{j, n} := \frac{\mu_{j, n}^{-s} \sum_{i} \lambda_{i, n} |U_{i, j}|^2}{\sum_{i, j}\lambda_{i, n} \mu_{j, n}^{-s} |U_{i, j}|^2}.
\end{equation}
Clearly we can upper bound the function $|f_n''(s)|$ by maximizing over the $\mu_{j, n}$ while keeping the probabilities $t_{j, n}$ fixed. We extend the set of allowed $\mu_{j, n}$ even more and consider all probability distributions for which $\mu_{j, n} \geq \frac{\tau}{1 + \tau} \lambda_{\min}(\sigma)^{n - m - r}$. We are hence interested in maximizing the function
\begin{equation}
g(\mu_{1, n}, \mu_{2, n},  ...) = \frac{1}{n} \left( \sum_{j} t_{j, n} (\log \mu_{j, n})^2 - \left( \sum_{j} t_{j, n} \log \mu_{j, n}   \right)^2 \right)
\end{equation}
over the set of probability distributions $\{ \mu_{j, n} \}$ such that
\begin{equation} \label{setset}
\mu_{j, n} \geq \frac{\tau}{1 + \tau} \lambda_{\min}(\sigma)^{n - m - r},
\end{equation}
for all $j$.

The function $g$ will reach its maximum either on its extreme points or on the boundary of the set in which the maximization is performed. A simple calculation gives
\begin{equation}
\frac{\partial g}{\partial \mu_{k, n}} = \frac{1}{n} \left( 2 t_{k, n} \frac{\log \mu_{k, n}}{\mu_{k, n}} - 2 \left( \sum_{j} t_{j, n} \log \mu_{j, n} \right) \frac{t_{k, n}}{\mu_{k, n}} \right) = 0 \Rightarrow \log \mu_{k, n} = \sum_i t_{i, n} \log \mu_{i, n}.
\end{equation}
Hence, in the extreme points of $g$ all the $\mu_{k, n}$ are equal and it is then easy to see that $g(\mu, \mu, ...) = 0$. As $g$ is positive, it then follows that the maximum of $g$ is attained on the boundary of the set in which the maximization is performed. Such boundary is composed of subsets of the original set given by Eq. (\ref{setset}) in which at least one of the $\mu_{j, n}$ is equal to $\frac{\tau}{1 + \tau} \lambda_{\min}(\sigma)^{n - m - r}$. Setting $\mu_{k, n} = \frac{\tau}{1 + \tau} \lambda_{\min}(\sigma)^{n - m - r}$, the new function to be maximized is 
\begin{equation}
\tilde g(\mu_{1, n},  ..., \mu_{k - 1, n}, \mu_{k + 1, n}, ...) = \frac{1}{n} \left( \sum_{j} t_{j, n} (\log \mu_{j, n})^2 - \left( \sum_{j} t_{j, n} \log \mu_{j, n}   \right)^2 \right),
\end{equation}
where now $\mu_{k, n} = \frac{\tau}{1 + \tau} \lambda_{\min}(\sigma)^{n - m - r}$ is a constant. Proceeding exactly as before, we find again that all the extreme points of $\tilde g$ are again minima of the function and, hence, the maximum of $\tilde g$ is attained once more on the boundary of the the set of probabilities allowed. This, in turn, is given by the union of subsets of the set given by Eq. (\ref{setset}) in which at least two of the $\mu_{k, n}$ are equal to $\frac{\tau}{1 + \tau} \lambda_{\min}(\sigma)^{n - m - r}$. We can continue with this process to show that all $\mu_{k, n}$ except one are equal to $\frac{\tau}{1 + \tau} \lambda_{\min}(\sigma)^{n - m - r}$. We hence find that the optimal choice of parameters is given by
\begin{equation}
\begin{cases}
\tilde \mu_{j, n} = \frac{\tau}{1 + \tau} \lambda_{\min}(\sigma)^{n - m - r} & \text{if} \hspace{0.1 cm} j \neq k, \\
\tilde \mu_{k, n} = 1 + \frac{\tau}{1 + \tau}\lambda_{\min}(\sigma)^{n - m - r} - \frac{\tau}{1 + \tau}\lambda_{\min}(\sigma)^{n - m - r}, & \text{otherwise}
\end{cases}
\end{equation}
for some integer $k$. Let 

$M := \frac{\tau}{1 + \tau} \lambda_{\min}(\sigma)^{n - m - r}$ and $N := 1 + \frac{\tau}{1 + \tau}\lambda_{\min}(\sigma)^{n - m - r} - \frac{\tau}{1 + \tau}\lambda_{\min}(\sigma)^{n - m - r}$. It then follows that
\begin{eqnarray} \label{Eqtkn}
g(\tilde \mu_{1, n}, \tilde \mu_{2, n}, ...) &=& \frac{1}{n} \left(  (1 - t_{k, n})t_{k, n} \left( \log M \right)^2 + t_{k, n}\left( \log N  \right)^2 \right. \nonumber \\ &-& \left. t_{k, n}^2 \left( \log   N   \right)^2 -  2 t_{k, n}(1 - t_{k, n}) \left( \log  M  \log  N \right)   \right)
\end{eqnarray}

We have
\begin{equation} \label{boundontauand1plustau}
\left | \log M \right|, \left |  \log N   \right| \leq 2 \log(\lambda_{\min}^{-1}(\sigma)) n,
\end{equation}
for sufficiently large $n$, and
\begin{eqnarray}
t_{k, n} &=& \frac{\mu_{k, n}^{-s} \sum_{i} \lambda_{i, n} |U_{i, k}|^2}{\sum_{i, j}\lambda_{i, n} \mu_{j, n}^{-s} |U_{i, j}|^2} \nonumber \\ &\leq&  \frac{\lambda_{\max}(\pi_n) \sum_{i}|U_{i, k}|^2}{(\tau/((1 + \tau)D^n))^s \sum_{i, j}\lambda_{i, n}  |U_{i, j}|^2} \nonumber \\ &=& \lambda_{\max}(\pi_n) \left(  \frac{(1 + \tau)\lambda_{\min}(\sigma)^{-n + m + r}}{\tau} \right)^s, 
\end{eqnarray}
where the second inequality follows from $1 \geq \mu_{j, n} \geq \frac{\tau}{1 + \tau} \lambda_{\min}(\sigma)^{n - m - r}$, which is a direct consequence of Eq. (\ref{nonameuntilnow}). 

From Eq. (\ref{supconv}), we have the bound
\begin{equation}
\lambda_{\max}(\pi_n) \leq 2^{o(n)} \lambda_{\max}(\sum_{i}p_i \rho_j) \leq 2^{o(n)} \lambda_{\max}(\rho)^{n - o(n)}.
\end{equation} 
Thus
\begin{equation}
t_{k, n} \leq 2^{o(n)}\left(  \frac{(1 + \tau)}{\tau} \right)^s    (\lambda_{\min}(\sigma)^{-s} \lambda_{\max}(\rho))^n  \lambda_{\max}(\rho)^{-o(n)}. 
\end{equation}
As by assumption $\lambda_{\max}(\rho) < 1$, choosing $s <  \log( \lambda_{\max}(\rho))/ \log(\lambda_{\min}(\sigma))$, we get that for $n$ sufficiently large, $t_{k, n} \leq  (10 \log \lambda_{\min}^{-1}(\sigma)n)^{-1}$. Then, from Eqs. (\ref{Eqtkn}) and (\ref{boundontauand1plustau}), 
\begin{eqnarray}
g(\tilde \mu_{1, n}, \tilde \mu_{2, n}, ...) &\leq& 2 \log \lambda_{\min}^{-1}(\sigma)  n \left((1 - t_{k, n})t_{k, n} + t_{k, n} + t_{k, n}^2 + 2(1 - t_{k, n})t_{k, n} \right) \nonumber \\ &\leq& 10 \log \lambda_{\min}^{-1}(\sigma) t_{k, n} \leq 1,
\end{eqnarray}
and we are done.
\end{proof}

\vspace{0.8 cm}

The final lemma of this section relates the entropy of an almost power state along $\rho$ with its own entropy. 

\begin{lemma} \label{nonlockentropyalmostppowerstates}
Let $\pi_n$ be given by Eq. (\ref{pipipipin}) with $k, r = o(n)$. Then
\begin{equation}
S(\pi_n) \leq n S(\rho) + o(n).
\end{equation}
\end{lemma}

\begin{proof}
Let $\rho = \sum_{i=1}^d p_i \ket{i}\bra{i}$, with $d = \text{rank}(\rho)$, and
\begin{equation}
\rho^{\otimes n} := \sum_{i^n}p_{i^n}\ket{i^n}\bra{i^n}
\end{equation}
with $i^{n} := i_1 ... i_n$, $p_{i^n} := p_{i_1}...p_{i_n}$, and $\ket{i^n} := \ket{i_1} ... \ket{i_n}$. For $\delta > 0$ define the set of typical sequences by ${\cal T}_{\delta}^n := \{ i^n : | - \log p_{i^n} - n S(\rho) | \leq n \delta \}$, and the typical projector by 
\begin{equation}
\Pi_{\delta}^{n} := \sum_{i^n \in {\cal T}_{\delta}^n} \ket{i^n}\bra{i^n}.
\end{equation}
Then from e.g. \cite{HOW} (appendix C) we have 
\begin{equation} \label{ytytytytytytyuuyuyiuiu}
\tr(\rho^{\otimes n} \Pi_{\delta}^n) \geq 1 - e^{-b \delta^2 n}, 
\end{equation}
and 
\begin{equation} \label{ksahflkdsjhfpoiuweyr,xznvc}
\Pi_{\delta}^n \rho^{\otimes n} \Pi_{\delta}^n \geq 2^{-n(S(\rho) + \delta)} \Pi_{\delta}^n.
\end{equation}

Let $\Pi'_n := (\id^{\otimes r} \otimes \Pi_{n^{-1/4}}^{n-m-r}) \otimes \id_E$, where the first identity is applied to the first $r$ register of ${\cal H}^{\otimes n - m}$, while the second is applied to the purifying Hilbert space${\cal H}_E^{\otimes n - m}$. Writing $\ket{{\Phi}_{n, m, r}}$ as in Eq. (\ref{,mxzbcv,mxzbvc,mzcvvz}), we can define
\begin{equation} \label{dfadfasdf}
\ket{{\Phi'}_{n, m, r}} = c\Pi'_n\ket{\phi} + \sqrt{1 - c^2}e^{i\vartheta}\ket{\phi^{\bot}}
\end{equation}
and follow the argument in the proof of Lemma \ref{trickwithsoivadfsadfsadf} (which applies unchanged to $\ket{{\Phi'}_{n, m, r}}$) to get 
\begin{eqnarray}
\Pi_{n^{-1/4}}^{n-m-r}\tr_E\left((\ket{\theta}\bra{\theta})^{\otimes n - m - r}\right)\Pi_{n^{-1/4}}^{n-m-r} \leq 2^{nh\left(\frac{r}{n - m}\right)}n^2 \tr_{1, ... r}\tr_E(\ket{{\Phi'}_{n, m, r}}\bra{{\Phi'}_{n, m, r}}).
\end{eqnarray}
Hence from Eq. (\ref{ksahflkdsjhfpoiuweyr,xznvc}),
\begin{equation} \label{eigsafsadfas}
\lambda_{\min}\left(\tr_{1, ... r} \tr_E \left(\ket{{\Phi'}_{n, m, r}}\bra{{\Phi'}_{n, m, r}})\right)\right) \geq 2^{o(n)} \lambda_{\min}(\Pi_{n^{-1/4}}^n \rho^{\otimes n} \Pi_{n^{-1/4}}^n ) \geq  2^{-n(S(\rho) + o(n))}.
\end{equation}

Moreover, Eqs. (\ref{ytytytytytytyuuyuyiuiu}) and (\ref{dfadfasdf}) give
\begin{eqnarray} \label{sdfsdfkjhsafd;lkh;khfsdfafd}
|\braket{{\Phi'}_{n, m, r}}{{\Phi}_{n, m, r}}| &=& c^2 \bra{\phi}\Pi_{n}'\ket{\phi} + (1 - c^2) \nonumber \\ &=& c^2 \tr(\rho^{\otimes n - m - r} \Pi_{n^{-1/4}}^{n - m - r}) + (1 - c^2) \nonumber \\ &\geq& 1 - e^{-n^{1/8}}, 
\end{eqnarray}
for sufficiently large $n$. Defining, 
\begin{equation} 
\pi'_n := \tr_{1,...,r}\tr_E(\ket{\Phi'_{n, m, r}}\bra{\Phi'_{n, m, r}}) / \braket{\Phi'_{n, m, r}}{\Phi'_{n, m, r}},
\end{equation}
we get from Eq. (\ref{sdfsdfkjhsafd;lkh;khfsdfafd}) that 
\begin{equation} \label{appppproooxxx}
|| \pi_n - \pi'_n ||_1 = o(1).
\end{equation}
Furthermore, from Eq. (\ref{eigsafsadfas}), $\lambda_{\min}(\pi'_n) \geq 2^{-n(S(\rho) + o(n))}$, and thus 
\begin{equation} \label{fiiispdfispofdsd}
S(\pi'_n) \leq - \log \lambda_{\min}(\pi'_n) \leq n S(\rho) + o(n).
\end{equation}
The lemma follows from Eqs. (\ref{appppproooxxx}), (\ref{fiiispdfispofdsd}) and Fannes inequality \cite{Fannes}.
\end{proof}

\section{Proof of Corollary \ref{faithful}} \label{coroll}

In this section we prove that the regularized relative entropy of 
entanglement is faithful. The idea is to combine Theorem 
\ref{maintheorem} with the exponential de Finetti theorem \cite{Ren05, Ren07}. 

\vspace{0.5 cm}

\begin{proof} (Corollary \ref{faithful})

In the following paragraphs we prove that for every entangled state 
$\rho \in {\cal D}({\cal H}_1 \otimes ... \otimes {\cal H}_m)$, there is a 
$\mu(\rho) > 0$ and a sequence of POVM elements 
$0 \leq A_n \le \id$, where $A_n$ acts on $({\cal H}_1 \otimes ... \otimes {\cal H}_m)^{\otimes n}$, such that
\begin{equation}
        \lim_{n \rightarrow \infty} \tr(A_n \rho^{\otimes n}) = 1, 
\end{equation}
and for all sequences of separable states 
$\{ \omega_n \}_{n \in \mathbb{N}}$,
\begin{equation}
        - \frac{\log \tr(A_n \omega_n)}{n} \geq \mu(\rho),
\end{equation}
From Theorem \ref{maintheorem} it will then follows that 
$E_R^{\infty}(\rho) \geq \mu(\rho) > 0$ (actually we only need Corollary \ref{maincompactstrongconverse} here).

The $A_n$'s are defined as follows. We apply the symmetrization operation $\hat{S}_n$ to 
the $n$ individual Hilbert spaces, trace out the first $\alpha n$ systems ($0 < \alpha < 1$), and then measure a LOCC informationally complete POVM $\{ M_k \}_{k=1}^{L}$ in each of the remaining $(1 - \alpha)n$ systems, obtaining an empirical frequency distribution $p_{k, n}$ of the possible 
outcomes $\{ k \}_{k=1}^L$ (see Appendix \ref{infcompPOVM}). Using this probability distribution, we form the operator
\begin{equation}
        L_n := \sum_{k=1}^L p_{k, n} M_k^{*},
\end{equation}
where $\{ M_k^* \}$ is the dual set of the family $\{ M_k \}$. If 
\begin{equation}
        || L_n - \rho ||_1 \leq \epsilon/2,
\end{equation}
where
\begin{equation} \label{epsfar}
        \epsilon := \min_{\sigma \in {\cal S}} || \rho - \sigma ||_1 > 0,
\end{equation}
we accept, otherwise we reject. Then we set 
$A_n := \hat{S}_n(\id^{\otimes \alpha n} \otimes \tilde{A}_n)$ as 
the POVM element associated to the event that we accept, where 
$\tilde{A}_n$ is the POVM element associated to measuring 
$\{ M_k \}_{k=1}^{L}$ on each of the $(1 - \alpha)n$ copies and 
accepting. 

First, by the law of large numbers \cite{Dud02} and the definition 
of informationally complete POVMs, it is clear that 
$\lim_{n \rightarrow \infty} \tr(A_n \rho^{\otimes n}) = 1$. It 
thus remains to show that $\tr(A_n \omega_n) = \tr(\id^{\otimes 
\alpha n} \otimes \tilde{A}_n) \hat{S}_n(\omega_n))  \leq 
2^{- \mu n}$, for a positive number $\mu$ and every sequence 
of separable states $\{ \omega_n \}_{n \in \mathbb{N}}$.  

Applying Theorem \ref{expdefinetti} with $k = \alpha n$ and 
$r = \beta n$ to $\tr_{1,...,\alpha n}(\hat{S}_n(\omega_n))$, 
we find that there is a probability measure $\nu$ such that
\begin{equation} \label{neweqcorr} 
        \tr_{1,...,\alpha n}(\hat{S}_n(\omega_n)) = 
        \int_{\sigma \in D({\cal H})} \int_{\ket{\theta} 
        \supset \sigma} \nu(d\ket{\theta}) \pi_{n}^{\ket{\theta}} + X_n, 
\end{equation}
where $|| X_n ||_1 \leq 2^{\frac{\alpha \beta n}{3}}$ for sufficiently large $n$,
\begin{equation}
        \pi_n^{\ket{\theta}} := \tr_{E}\left( \ket{\psi^{\ket{\theta}}_{(1-\alpha)n}} \bra{\psi^{\ket{\theta}}_{(1-\alpha)n}} \right),
\end{equation}
and $\ket{\psi^{\ket{\theta}}_{(1-\alpha)n}} \in \ket{\theta}^{[\otimes, (1-\alpha)n, \beta n]}$. 

In the next paragraphs we show that only an exponentially small 
portion of the volume of $\nu$ is in a neighborhood of 
purifications of $\rho$. 

Since we are measuring local POVMs, the operation $\pi \mapsto 
\tr_{\backslash 1}(\hat{S}_n(\pi) \id^{\otimes \alpha n} \otimes \tilde{A}_n)$ 
is a stochastic LOCC map (see e.g. \cite{HHHH07}). It hence follows from 
Eq. (\ref{neweqcorr}) that 
\begin{eqnarray} \label{eqcor1}
\tr_{\backslash 1}(\hat{S}_n(\omega_n) \id \otimes \tilde{A}_n) &=& \int_{\sigma \in B_{2\epsilon}(\rho)} \int_{\ket{\theta} \supset \sigma} \nu(d\ket{\theta}) \tr_{\backslash 1}(\pi_{n}^{\ket{\theta}}\id \otimes \tilde{A}_n) \nonumber \\ &+& \int_{\sigma \in \notin B_{2\epsilon}(\rho)} \int_{\ket{\theta} \supset \sigma} \nu(d\ket{\theta}) \tr_{\backslash 1}(\pi_{n}^{\ket{\theta}}\id \otimes \tilde{A}_n) \nonumber \\ &+& \tr_{\backslash 1}(X_n\id \otimes \tilde{A}_n) \in \text{cone}({\cal S}).
\end{eqnarray}
As $|| X_n || \leq 2^{-\alpha \beta n/3}$, we find $|| \tr_{\backslash 1}(X_n\id \otimes \tilde{A}_n) ||_1 \leq 2^{-\alpha \beta n/3}$. 

Furthermore, from Lemma \ref{hof} we have that if $\tr_E(\ket{\theta}\bra{\theta}) \notin B_{2\epsilon}(\rho)$,
\begin{equation}
|| \tr_{\backslash 1}(\pi_{n}^{\ket{\theta}}\id \otimes \tilde{A}_n) ||_1 = \tr(\pi_{n}^{\ket{\theta}}\id \otimes \tilde{A}_n) \leq n^{d^2}2^{-(\epsilon/K - h(\beta))(1 - \alpha)n},
\end{equation}
where $K$ is given by Eq. (\ref{eq3.1.5}) and can be taken to be such that $K \leq \dim({\cal H})^4$. 
  
Putting it all together,
\begin{eqnarray}
\tr_{\backslash 1}(\hat{S}_n(\omega_n) \id \otimes \tilde{A}_n) = \int_{\sigma \in B_{2\epsilon}(\rho)} \int_{\ket{\theta} \supset \sigma} \nu(d\ket{\theta}) \tr_{\backslash 1}(\pi_{n}^{\ket{\theta}}\id \otimes \tilde{A}_n) + \tilde{X}_n \in \text{cone}({\cal S}).
\end{eqnarray}
with $\tilde{X_n}$ given by the sum of the two last terms in Eq. (\ref{eqcor1}), which satisfies $|| \tilde{X}_n ||_1 \leq 2^{-\alpha \beta n/3} + n^{d^2}2^{-(\epsilon/K - h(\beta))(1 - \alpha)n}$.

For each $\tr_{\backslash 1}(\pi_{n}^{\ket{\theta}}\id \otimes \tilde{A}_n)$, with $\tr_E(\ket{\theta}\bra{\theta}) \in B_{2\epsilon}(\rho)$, we can write 
\begin{equation}
\tr_{\backslash 1}(\pi_{n}^{\ket{\theta}}\id \otimes \tilde{A}_n) = \tr_{\backslash 1}(\pi_{n}^{\ket{\theta}}\id \otimes B_n) + \tr_{\backslash 1}(\pi_{n}^{\ket{\theta}}\id \otimes (\tilde{A}_n - B_n)),
\end{equation}
where $B_n$ is the sum of the POVM elements for which the post-selected state is $\delta$-close from the empirical state. 

From Lemma \ref{KRKRKRKKR} we find that $\tr(\pi_{n}^{\ket{\theta}} \id \otimes (\tilde{A_n} - B_n)) \leq 2^{-M (1 - \alpha)\delta^2 n}$. Therefore, 
\begin{eqnarray} \label{hatx}
\tr_{\backslash 1}(\hat{S}_n(\omega_n) \id \otimes \tilde{A}_n) &=& \int_{\sigma \in D({\cal H})} \int_{\ket{\theta} \supset \sigma \in B_{2\epsilon}(\rho)} \nu(d\ket{\theta}) \tr(\pi_{n}^{\ket{\theta}}\id \otimes B_n)\rho^{\ket{\theta}} \nonumber \\ &+& \hat{X}_n \in \text{cone}({\cal S}).
\end{eqnarray}
where $\hat{X}_n$ is such that $|| \hat{X}_n ||_1 \leq 2^{-\alpha \beta n/3} + n^{d^2}2^{-(\epsilon/K - h(\beta))(1 - \alpha)n} + 2^{-M (1 - \alpha)\delta^2 n}$ and
\begin{equation}
\rho^{\ket{\theta}} := \frac{\tr_{\backslash 1}(\pi_{n}^{\ket{\theta}} \id \otimes B_n)}{\tr(\pi_{n}^{\ket{\theta}}\id \otimes B_n)}.
\end{equation}
Note that we have $|| \rho^{\ket{\theta}} - \rho|| \leq \delta + \epsilon/2$ for every $\rho^{\ket{\theta}}$ appearing in the integral of Eq. (\ref{hatx}). Define 
\begin{equation}
\Lambda := \int_{\sigma \in D({\cal H})} \int_{\ket{\theta} \supset \sigma \in B_{2\epsilon}(\rho)} \nu(d\ket{\theta}) \tr(\pi_{n}^{\ket{\theta}}\id \otimes B_n). 
\end{equation}
Then, 
\begin{equation} \label{wehereweused}
\left \Vert  \Lambda^{-1}\int_{\sigma \in D({\cal H})} \int_{\ket{\theta} \supset \sigma \in B_{2\epsilon}(\rho)} \nu(d\ket{\theta}) \tr(\pi_{n}^{\ket{\theta}}\id \otimes B_n)\rho^{\ket{\theta}} - \rho \right \Vert \leq \delta + \epsilon/2, 
\end{equation}
From Eqs. (\ref{epsfar}) and (\ref{wehereweused}) it follows that $\Lambda^{-1}\int_{\sigma \in D({\cal H})} \int_{\ket{\theta} \supset \sigma \in B_{2\epsilon}(\rho)} \nu(d\ket{\theta}) \tr(\pi_{n}^{\ket{\theta}}\id \otimes B_n)\rho^{\ket{\theta}}$ is at least $\epsilon/2 - \delta$ far away from the separable states set. Using Eq. (\ref{hatx}) we thus find that 
\begin{equation}
\Lambda \leq (\epsilon/2 - \delta)^{-1}(2^{-\alpha \beta n/3} + n^{d^2}2^{-(\epsilon/K - h(\beta))n} + n2^{-((1 - \alpha)n - 1)\delta^2M^{-2}}).
\end{equation}
With this bound we finally see that
\begin{eqnarray}
\tr(\omega_n A_n) &=& \tr(\hat{S}_n(\omega_n) \id \otimes \tilde{A}_n) \nonumber \\ &=& \Lambda + \tr(\hat{X}) \nonumber \\ &\leq& (1 + (\epsilon/2 - \delta)^{-1})(2^{-\alpha \beta n/3} + n^{d^2}2^{-(\epsilon/K - h(\beta))n} + n2^{-((1 - \alpha)n - 1)\delta^2M^{-2}}) \nonumber \\ &\leq& 2^{- \mu n},
\end{eqnarray}
for appropriately chosen $\alpha, \beta \in [0, 1]$ and $\mu > 0$.
\end{proof}
\vspace{0.3 cm}

In the proof above the only property of the set of separable states that we used, apart from the five properties required for Theorem \ref{maintheorem} to hold, was its closedness under SLOCC. It is an interesting question if such a property is really needed, or if actually the positiveness of the rate function is a generic property of any $\rho \notin {\cal M}$ for every family of sets satisfying Theorem \ref{maintheorem}. The following example shows that this is not the case; for some choices of sets $\{ {\cal M}_k \}$ the rate function can be zero for a state $\rho \notin {\cal M}$. In fact, in our example the rate function is zero for every state.   

A bipartite state $\sigma_{AB}$ is called $n$-extendible if there is a state $\tilde{\sigma}_{AB_1...B_n}$ symmetric under the permutation of the $B$ systems and such that $\tr_{B_2,...,B_n}(\tilde{\sigma}) = \sigma$. Let us denote the set of $n$-extendible states acting on ${\cal H} = {\cal H}_A \otimes {\cal H}_B$ by ${\cal E}_k({\cal H})$. It is clear that the sets $\{ {\cal E}_k({\cal H}^{\otimes n}) \}_{n \in \mathbb{N}}$ satisfy conditions \ref{cond1}-\ref{cond5} and therefore we can apply Theorem \ref{maintheorem} to them. Corollary \ref{faithful} however does not hold in this case, as the sets are not closed under two-way LOCC, even though they are closed under one-way LOCC. In fact, the statement of the corollary fails dramatically in this case as it turns out that the measures $E_{{\cal E}_k}^{\infty}$ are zero for every state. This can be seen as follows: Given a state $\rho$, let us form the $k$-extendible state
\begin{equation}
\tilde{\rho}_{AB_1,...,B_k} := \id_A \otimes \hat{S}_{B_1,...,B_k}\left(\rho_{AB} \otimes \left(\frac{\id}{d^2}\right)^{\otimes k - 1} \right)
\end{equation}
We have $\tilde{\rho}_{AB_1,...,B_k} \geq \rho_{AB} \otimes \frac{\id}{d^2}^{\otimes k - 1}/k$. Then, from the operator monotonicity of the $\log$,
\begin{equation}
E_{{\cal E}_k}(\rho) \leq S(\rho || \tr_{B_2,...,B_n}(\tilde{\rho})) \leq k. 
\end{equation} 
As the upper bound above is independent of $n$, we then find
\begin{equation} 
E_{{\cal E}_k}^{\infty}(\rho) = \lim_{n \rightarrow \infty} \frac{1}{n} E_{{\cal E}_k}^{\infty}(\rho^{\otimes n}) \leq \lim_{n \rightarrow \infty} \frac{k}{n} = 0. 
\end{equation}

Note that as ${\cal E}_1$ is contained in the set of one-way undistillable states ${\cal C}_{\text{one-way}}$, the same is true for $E_{{\cal C}_{\text{one-way}}}^{\infty}$, i.e. it is identically zero. It is interesting that an one-way distillable state cannot be distinguished with an exponential decreasing probability of error from one-way undistillable states if we allow these to be correlated among several copies, while any entangled state can be distinguished from arbitrary sequences of separable states with exponentially accuracy. Moreover, as the set of states with a positive partial transpose (PPT) satisfy conditions \ref{cond1}-\ref{cond5} and is closed under SLOCC, every state with a non-positive partial transpose (NPPT) can be exponentially well distinguished from a sequence of PPT states. It is an intriguing open question if the same holds for distinguishing a two-way distillable state from a sequence of two-way undistillable states. Due to the conjecture existence of NPPT bound (undistillable) entanglement \cite{DSS+00, DCLB00, Cla06, BE08}, property \ref{cond4} might fail and therefore we do not know what happens in this case. 

\section{Proof of Corollary \ref{positiverate}}

\begin{proof} (Corollary \ref{positiverate})

The proof is a simple application of the well-known idea of bounding the rate of asymptotic entanglement transformations by entanglement measures (see e.g. \cite{PV07, HHHH07}). Suppose we can transform $\rho$ into $\sigma$ asymptotically, where $\sigma$ is entangled. Then, for every $\epsilon > 0$ there is a sequence of LOCC maps $\{ \Lambda_n \}_{n \in \mathbb{N}}$ and a sequence of integers $\{ k_n \}_{n \in \mathbb{N}}$ such that
\begin{equation}
\lim_{n \rightarrow \infty} || \Lambda_n(\rho^{\otimes k_n}) - \sigma^{\otimes n} ||_1 = 0.
\end{equation}
and
\begin{equation}
\limsup_{n \rightarrow \infty} \frac{k_n}{n} \leq R(\rho \rightarrow \sigma) + \epsilon.
\end{equation}
From the monotonicity of the relative entropy of entanglement under LOCC \cite{VP98} and its asymptotically continuity (see Lemma \ref{asympcont}), we find 
\begin{eqnarray}
E_R^{\infty}(\sigma) &=& \limsup_{n \rightarrow \infty} \frac{1}{n}E_R(\sigma^{\otimes n}) \nonumber \\ &=& \limsup_{n \rightarrow \infty} \frac{1}{n}E_R(\Lambda_n(\rho^{\otimes k_n})) \nonumber \\ &\leq& \limsup_{n \rightarrow \infty} \frac{1}{n}E_R(\rho^{\otimes k_n}) \nonumber \\ &=& \limsup_{n \rightarrow \infty} \frac{k_n}{n} \limsup \frac{1}{k_n}E_R(\rho^{\otimes k_n}) \nonumber \\ &\leq& (R(\rho \rightarrow \sigma) + \epsilon)E_R^{\infty}(\rho).
\end{eqnarray}
As, from Corollary \ref{faithful}, $E_R^{\infty}(\sigma) > 0$ and $\epsilon > 0$ is arbitrary, we find that indeed $R(\rho \rightarrow \sigma) > 0$.
\end{proof}

\section{Acknowledgments}

We gratefully thank Koenraad Audenaert, Nilanjana Datta, Jens Eisert, Andrzej Grudka, Masahito Hayashi, Micha{\l} and Ryszard Horodecki, Renato Renner, Shashank Virmani, Reinhard 
Werner, Andreas Winter and the participants in the 2009 McGill-Bellairs workshop for many interesting discussions, and an anonymous referee for filling in gaps in the proofs of Lemma \ref{ogawanagaokaadaptation} and Proposition \ref{maincompact}, for pointing out that our main result could be extended to cover the original quantum Stein's Lemma and for many other extremely useful comments on the manuscript. This work is 
part of the QIP-IRC supported by EPSRC (GR/S82176/0) as well as the 
Integrated Project Qubit Applications (QAP) supported by the IST directorate 
as Contract Number 015848' and was supported by the Brazilian agency Funda��o de Amparo � Pesquisa do Estado de Minas Gerais (FAPEMIG), 
an EPSRC Postdoctoral Fellowship for Theoretical Physics and a Royal Society Wolfson Research Merit Award.\\

\appendix

\section{Informationally Complete POVMs} \label{infcompPOVM}

An informationally complete POVM in ${\cal B}(\mathbb{C}^{m})$ is defined as a set of positive semi-definite operators $A_i$ forming a resolution of the identity and such that $\{ A_i \}$ forms a basis for ${\cal B}(\mathbb{C}^{m})$. Informationally complete POVMs can be explicitly constructed in every dimension (see e.g. \cite{KR05}).

We say that a family $\{ M_i \}$  of elements from 
${\cal B}(\mathbb{C}^m)$ is a dual of the a 
family $\{ M_i^* \}$ if for all $X \in {\cal B}(\mathbb{C}^m)$,
\begin{equation} \label{dualbasis}
X = \sum_{i} \tr[M_i X] M_i^*.
\end{equation}
The above equation implies in particular that the operator 
$X$ is fully determined by the expectations values $\tr[M_i X]$. 
Another useful property is that for every informationally complete POVM in ${\cal B}(\mathbb{C}^{m})$ there is a real number $K_m$ such that for every two states $\rho$ and $\sigma$,
\begin{equation} \label{eq3.1.5}
|| \rho - \sigma ||_1 \leq K_m || p_{\rho} - p_{\sigma}||_1,
\end{equation}
with $p_{\rho} = \tr(M_i \rho)_i$ and $p_{\sigma} = \tr(M_i \sigma)_i$. For example, in the family of informationally complete POVM constructed in Ref. \cite{KR05}, $K_m \leq m^4$.
                                                                              
\section{Exponential Quantum de Finetti Theorem} \label{Rennerexp}               
                                                                                
There have been several interesting recent developments on quantum 
versions \cite{KR05, Ren05, CKMR07, Ren07} of the seminal result by 
Bruno de Finetti on the characterization of exchangeable probability 
distributions \cite{dFin37}. Here we state an 
exponential version of the theorem for quantum states, recently proved 
by Renner \cite{Ren05, Ren07}.  

\begin{theorem} \label{expdefinetti}
\cite{Ren05,Ren07, KM07} For any state $\ket{\psi_{n + k}} \in 
\text{Sym}({\cal H}^{\otimes n + k})$ there exists a measure $\mu$ 
over ${\cal H}$ and for each pure state $\ket{\theta} \in {\cal H}$ 
another pure state $\ket{\psi^{\theta}_n} \in \ket{\theta}^{[\otimes, n, r]}$ 
such that
\begin{equation}
        \left \Vert \tr_{1, ..., k}(\ket{\psi_{n+k}}\bra{\psi_{n+k}})         
        - \int \mu(d\ket{\theta}) \ket{\psi^{\theta}_{n}}\bra{\psi^{\theta}_{n}} \right \Vert_1 \leq n^{\dim({\cal H})} 2^{- \frac{k(r + 1)}{2(n + k)}}.
\end{equation}
\end{theorem}

The generalization of Theorem \ref{expdefinetti} to permutation-symmetric 
mixed states goes as follows. First, we use the fact that every 
permutation-symmetric mixed state $\rho_{n + k}^S$ acting on 
${\cal H}_S^{\otimes n + k}$ has a symmetric purification 
$\ket{\psi}^{SE}_{n + k} \in ({\cal H}_S \otimes {\cal H}_E)^{\otimes n + k}$, 
with $\dim({\cal H}_E) = \dim({\cal H}_S)$ (see e.g. Lemma 4.2.2 of Ref. 
\cite{Ren05}). Then we apply Theorem \ref{expdefinetti} to 
$\ket{\psi}^{SE}_{n + k}$ and use the contractiveness of the trace 
norm under the partial trace to find 
\begin{equation} \label{df1}
        \left \Vert \tr_{1,...,k}(\rho_{n + k}) - \int \mu(d\sigma)
        \rho_{\sigma} \right \Vert_1 \leq n^{\dim({\cal H})^2} 
        2^{- \frac{k(r + 1)}{2(n + k)}}
\end{equation}
where
\begin{equation} \label{df2}
        \rho_{\sigma} := \tr_{E}(\ket{\psi_{n}^{\ket{\theta}}}
        \bra{\psi_{n}^{\ket{\theta}}}),
\end{equation}
with $\sigma := \tr_E(\ket{\theta}\bra{\theta})$ and
\begin{equation} \label{df3}
        \mu(d\sigma) := \int_{\ket{\theta} \supset \sigma} \mu(d\ket{\theta}).
\end{equation}
In the equation above $\ket{\theta} \supset \sigma$ means that the 
integration is taken with respect to the purifying system $E$ and 
runs over all purifications of $\sigma$.

\subsubsection{Chernoff-Hoeffding Bound for Almost Power States}

The states $\tr_{E}(\ket{\psi^{\theta}_n}\bra{\psi^{\theta}_n})$ behave 
like $\tr_E(\ket{\theta}\bra{\theta})^{\otimes n}$ in many respects. One 
example is the case where the same POVM is measured on all the $n$ copies. 

Let $\{ M_{\omega} \}_{\omega \in {\cal W}}$ be a POVM on ${\cal H}$ and define its induced probability distribution on $\ket{\theta}$ by $P_{M}(\ket{\theta}\bra{\theta}) = \{ \bra{\theta}M_{\omega}\ket{\theta} \}_{\omega \in {\cal W}}$. Theorems 4.5.2 of Ref. \cite{Ren05} and its reformulation as Lemma 2 of Ref. \cite{HHH+08a} show the following.
\begin{lemma} \label{hof}
\cite{Ren05, HHH+08a} Let $\ket{\Psi_n}$ be a vector from $\ket{\theta}^{[\otimes, n, r]}$ with $0 \leq r \leq \frac{n}{2}$ and $\{ M_{\omega} \}_{\omega \in {\cal W}}$ be a POVM on ${\cal H}$.  
\begin{equation} \label{testtest}
Pr \left(  \Vert P_M(\ket{\theta}\bra{\theta}) - P_M(\ket{\Psi_n}\bra{\Psi_n}) \Vert_1 > \delta \right) \leq 2^{- n \left(\frac{\delta^2}{4} - h\left(\frac{r}{n}  \right)  \right) + | {\cal W} |\log(\frac{n}{2} + 1)} 
\end{equation}
where $P_M(\ket{\Psi_n}\bra{\Psi_n})$ is the frequency distribution of outcomes of $M^{\otimes n}$ applied to $\ket{\Psi_n}\bra{\Psi_n}$, and the probability is taken over those outcomes. 
\end{lemma}

This Lemma shows that apart from the factor $h(r/n)$, which in an usual application of Lemma \ref{hof} is taken to be vanishing small, the statistics of the frequency distribution obtained by measuring an almost power state along $\ket{\theta}$ is the same as if we had $\ket{\theta}^{\otimes n}$. 

\subsection{Post-selected states}

The next lemma, due to K\"onig and Renner, appeared in \cite{KR05} as Theorem A.1 and is used in the proof of Corollary \ref{faithful}.

\begin{lemma} \label{KRKRKRKKR}
\cite{KR05} Let $\rho_{m + 1} \in {\cal D}({\cal H}^{\otimes m + 1})$ be a permutation-symmetric state and ${\cal M} := \{ M_k \}$ an informationally complete POVM in ${\cal H}$. Consider the probability distribution
\begin{equation}
        p(i_1,...,i_m) := \tr(\id 
        \otimes M_{i_{1}}\otimes M_{i_{2}} \otimes ... 
        \otimes M_{i_{m}}\rho_{m + 1}),
\end{equation}
associated to the measurement of ${\cal M}$ in $m$ of the subsystems of $\rho_{m + 1}$. Define the post-selected states
\begin{equation} \label{definet}
\pi_{i_1,...,i_m} := \frac{\tr_{\backslash 1}(\id \otimes  M_{i_{1}}\otimes M_{i_{2}} \otimes ... \otimes M_{i_{m}} \rho_{m + 1})}{\tr(\id \otimes M_{i_{1}}\otimes M_{i_{2}} \otimes ... \otimes M_{i_{m}}\rho_{m + 1})}
\end{equation}
and let $L_m^{i_1,...,i_{m}}$ be the estimated state when the sequence of outcome $\{i_1,...,i_{m}\}$ is obtained. Define ${\cal R}$ as the set of all outcome sequences such that 
\begin{equation}
|| L_m^{i_1,...,i_{m}} - \pi_{i_1,...,i_{m}} ||_1 \geq \delta.
\end{equation}
Then there is a $M > 0$ (only depending on the dimension of ${\cal H}$ and on the POVM ${\cal M}$) such that
\begin{equation} \label{KRen}
\sum_{(i_1,...,i_{m}) \in {\cal R}} p(i_1,...,i_{m}) \leq  2^{-M m \delta^2}.
\end{equation}
\end{lemma}

\section{Useful Results} \label{usefulresult}
Defining the fidelity 
$F(\rho,\sigma)= \tr(\sqrt{\sqrt{\rho}\sigma\sqrt{\rho}})$ we
find \cite{Uhl76}
\begin{lemma} \label{relation3}
For every $\rho, \sigma \in {\cal D}({\cal H})$,
\begin{equation}
        1 - F(\rho, \sigma) \leq \frac{1}{2} || \rho - \sigma  ||_1 
        = \tr(\rho - \sigma)_+ \leq \sqrt{1 - F(\rho, \sigma)^2}.
\end{equation}
\end{lemma}

\begin{lemma} \label{monotonicity3}
For $A, B$ positive semidefinite and $\Lambda$ a trace-preserving completely positive map,
\begin{equation}
||\Lambda(A)||_1 \leq || A ||_1, \hspace{0.5 cm} \tr(\Lambda(A))_+ \leq \tr(A)_+, \hspace{0.5 cm} F(\Lambda(A), \Lambda(B)) \geq F(A, B).
\end{equation}
\end{lemma}

Let $E: {\cal D}({\cal H}) \rightarrow \mathbb{R}_+$. We say $E$ is \textit{asymptotically continuous} if for every $\rho, \sigma \in {\cal D}({\cal H})$,
\begin{equation}
|E(\rho) - E(\sigma)| \leq \log(\dim({\cal H}))f(|| \rho - \sigma ||_1),
\end{equation}
for a real-valued function $f: \mathbb{R}_+ \rightarrow \mathbb{R}_+$ \textit{independent} of $\dim({\cal H})$ and such that $\lim_{x \rightarrow 0} f(x) = 0$. Although not strictly needed, we will also demand that $f$ is monotonic increasing, in order to simplify some of the proofs.  

The next Lemma is due to Synak-Radtke and Horodecki \cite{SH06} and Christandl \cite{Chr06}. 
\begin{lemma} \label{asympcont}
\cite{SH06, Chr06} For every family of sets $\{ {\cal M}_n \}_{n \in \mathbb{N}}$ satisfying properties \ref{cond1}-\ref{cond4}, $E_{{\cal M}_n}$ and $E_{{\cal M}}^{\infty}$, given by Eqs. (\ref{relent1}) and (\ref{regu1}), respectively, are asymptotically continuous. 
\end{lemma}

In Ref. \cite{SH06} it was shown that the minimum relative entropy over any convex set that includes the maximal mixed state is asymptotically continuous. It is simple to check that their proof goes through if instead of the maximally mixed state, the set contains $\sigma^{\otimes n}$, for a full rank state $\sigma$. For $E_{{\cal M}_n}$ the lemma then follows from properties \ref{cond1} and \ref{cond2}. In Proposition 3.23 of Ref. \cite{Chr06}, in turn, it was proven that $E_R^{\infty}$ is asymptotically continuous. It is straightforward to note that the proof actually applies to the regularized minimum relative entropy over any family of sets satisfying properties \ref{cond1}-\ref{cond4}. Moreover, the functions $f$ used in \cite{SH06} and \cite{Chr06} turn out to be monotonic increasing.

The next two lemmata will play an important role in the proof of Proposition \ref{relenteqrob}. The first, due to Ogawa and Nagaoka, appeared in Ref. \cite{ON00} as Theorem 1 and was the key element for establishing the strong converse of quantum Stein's Lemma.

\begin{lemma} \label{ON}
\cite{ON00} Given two quantum states $\rho, \sigma \in {\cal D}({\cal H})$ such that $\text{supp}(\rho) \subseteq \text{supp}(\sigma)$ and a real number $\lambda$, 
\begin{equation}
\tr(\rho^{\otimes n} - 2^{\lambda n}\sigma^{\otimes n})_+ \leq 2^{- n (\lambda s - \psi(s))},
\end{equation}
for every $s \in [0, 1]$. The function $\psi(s)$ is defined as 
\begin{equation}
\psi(s) := \log( \tr( \rho^{1 + s}\sigma^{-s})).
\end{equation}
\end{lemma}
Note that $\psi(0) = 0$ and $\psi'(0) = S(\rho || \sigma)$. Hence, if $\lambda > S(\rho || \sigma)$, $\tr(\rho^{\otimes n} - 2^{\lambda n}\sigma^{\otimes n})_+$ goes to zero exponentially fast in $n$.

The next Lemma, due to Datta and Renner \cite{DR08}, appeared in Ref. \cite{DR08} as Lemma 5 and is used in the proofs of Propositions \ref{relenteqrob} and \ref{maincompact}.

\begin{lemma} \label{DR}
\cite{DR08} Let $\rho \in {\cal D}({\cal H})$ and $Y, \Delta$ be positive semidefinite operators such that $\rho \leq Y + \Delta$ and $\tr(\Delta) < 1$. Then there exists a state $\tilde{\rho} \in {\cal D}({\cal H})$ such that 
\begin{equation} \label{EqLemma0293820938}
\tilde \rho \leq (1 - \tr(\Delta))^{-1} Y,
\end{equation}
and
\begin{equation}
F(\rho, \tilde \rho) \geq 1 - \tr(\Delta), \hspace{0.3 cm} || \rho - \tilde \rho ||_1 \leq 4 \sqrt{\tr(\Delta)}.
\end{equation}
\end{lemma}
\begin{proof}
Let $T := Y^{1/2}(Y + \Delta)^{-1/2}$, $\rho' := T\rho T^{\cal y}$ and set $\tilde \rho := \rho' / \tr(\rho')$. As $\rho \leq Y + \Delta$, we find
\begin{equation} 
\rho' = T \rho T^{\cal y} \leq  Y
\end{equation}
and hence 
\begin{equation}\label{eqlemmaDR1}
\tilde{\rho} = \tr(\rho')\rho' \leq \tr(T^{\cal y}T \rho) Y.
\end{equation}
Let us show that
\begin{equation} \label{eqpsdjldshfs}
\tr(T^{\cal y}T \rho) \geq 1 - \tr(\Delta). 
\end{equation}
Eq. (\ref{EqLemma0293820938}) then follows from Eqs. (\ref{eqlemmaDR1},\ref{eqpsdjldshfs}). Note that 
\begin{equation}
T^{\cal y}T = (Y + \Delta)^{-1/2}Y (Y + \Delta)^{-1/2} \leq \id.
\end{equation}
Then, using the inequality $\rho \leq Y + \Delta$,
\begin{equation}
\tr((\id - T^{\cal y}T)\rho) \leq \tr(Y + \Delta) - \tr((Y + \Delta)T^{\cal y}T) = \tr(\Delta),
\end{equation}
from which Eq. (\ref{eqpsdjldshfs}) follows. 

In the proof of Lemma 5 of Ref. \cite{DR08} it is proven that $F(\rho, \rho') \geq 1 - \tr(\Delta)$. Hence
\begin{equation}
F(\rho, \tilde \rho) = \tr(\rho')^{-1/2}F(\rho, \rho') \geq F(\rho, \rho') \geq 1 - \tr(\Delta),
\end{equation}
where we used that $\tr(\rho') = \tr(T^{\cal y}T \rho) \leq 1$, which follows from $T^{\cal y}T \leq \id$. The inequality for the trace norm follows from Eq. (\ref{relation3}).
\end{proof}

We also make use of the following simple lemma.
 
\begin{lemma} \label{supversusconvex}
Let $\ket{\Psi} \in {\cal H}$ be such that 
$\ket{\Psi} := \sum_{k \in {\cal X}} \ket{\psi_k}$. Then 
\begin{equation}
\ket{\Psi}\bra{\Psi} \leq |{\cal X}| \sum_{k \in {\cal X}} \ket{\psi_k}\bra{\psi_k}
\end{equation}
\end{lemma}
\begin{proof}
For every $\ket{\theta} \in {\cal H}$, $|\bra{\theta}(\ket{\psi_k}\bra{\psi_k'})\ket{\theta}|  = |\braket{\theta}{\psi_k}| |\braket{\theta}{\psi_k'}|$. Then,
\begin{eqnarray}
        \bra{\theta}\left( \ket{\Psi}\bra{\Psi}   \right)\ket{\theta} 
        &=& \left | \sum_{k, k'} \bra{\theta}\left( \ket{\psi_k}\bra{\psi_k'}\right)\ket{\theta} 
        \right | \nonumber \\
        & \leq& |{\cal X}|^2 \sum_{k, k'} \frac{1}{|{\cal X}|^2}\sqrt{\bra{\theta}
        \left( \ket{\psi_k}\bra{\psi_k}  \right)\ket{\theta}   \bra{\theta}\left( \ket{\psi_k'}
        \bra{\psi_k'}  \right)\ket{\theta}}   \nonumber \\
        &\leq&  |{\cal X}|^2  \sqrt{\sum_{k, k'}  \frac{1}{|{\cal X}|^2} \bra{\theta}\left( 
        \ket{\psi_k}\bra{\psi_k}  \right)\ket{\theta}   \bra{\theta}\left( \ket{\psi_k'}
        \bra{\psi_k'} \right)\ket{\theta}}  \nonumber \\
        &=& |{\cal X}| \bra{\theta} \left( \sum_{k \in {\cal X}} \ket{\psi_k}\bra{\psi_k}\right)
        \ket{\theta},
\end{eqnarray}
where the inequality in the third line follows from Jensen's inequality. 
\end{proof}

The final lemma, adapted from lemma 4.1.2 of \cite{Han03}, is used in the proof of Lemma \ref{ogawanagaokaadaptation}.
\begin{lemma} \label{hancover}
Given two probability distributions $p, q : \{1, ..., n \} \rightarrow \mathbb{R}$ and real numbers $0 \leq \lambda_i \leq 1$, $i \in \{1, ..., n  \}$, and $\mu$,
\begin{equation}
\sum_{i=1}^n \lambda_i (p(i) - 2^{\mu}q(i)) \leq \Pr_{\{ p \}} \left( i : \log \frac{p(i)}{q(i)} \geq \mu \right).
\end{equation}
\end{lemma}
\begin{proof}
The lemma can be proved by the following chain of inequalities
\begin{eqnarray}
\Pr_{\{ p \}} \left( i : \log \frac{p(i)}{q(i)} \geq \mu \right) &=& \sum_{i : p(i) \geq 2^{\mu}q(i)} p(i) \nonumber \\ &\geq& \sum_{i : p(i) \geq 2^{\mu}q(i)} \lambda_i p(i) \nonumber \\ &\geq& \sum_{i : p(i) \geq 2^{\mu}q(i)} \lambda_i (p(i) - 2^{\mu}q(i)) \nonumber \\ &\geq&  \sum_{i} \lambda_i (p(i) - 2^{\mu}q(i)).
\end{eqnarray}
In the first inequality we used that $0 \leq \lambda_i \leq 1$, in the second that $q(i) \geq 0$, and in the last that we add negative terms corresponding to the $i$'s for which $p(i) < 2^{\mu}q(i)$.
\end{proof}

\end{document}